\documentclass[reqno]{amsart}
\usepackage{amssymb}
\usepackage{url}
\usepackage{hyperref}
\allowdisplaybreaks[4]

\usepackage{mathrsfs}
\let\mathcal\mathscr
\usepackage[mathscr]{eucal}

\usepackage{color}

\let\phi=\varphi
\let\kappa=\varkappa
\DeclareMathOperator{\rank}{rank}
\DeclareMathOperator{\sym}{sym}
\DeclareMathOperator{\const}{const}

\newcommand*{\sd}[2]{\{\,#1\mid#2\,\}}

\newcommand*{\eval}[1]{\left.#1\right|}
\newcommand*{\abs}[1]{\left|#1\right|}
\newcommand*{\Ev}{\mathbf{E}}

\newcommand*{\dc}{\mathbf{X}}
\newcommand*{\ldb}{[\![}
\newcommand*{\rdb}{]\!]}

\theoremstyle{theorem}
\newtheorem{proposition}{Proposition}

\newtheorem{theorem}{Theorem}
\theoremstyle{definition}

\theoremstyle{remark}
\newtheorem{remark}{Remark}

\usepackage[all]{xy}
\SelectTips{cm}{}

\usepackage{mathrsfs}
\let\mathcal\mathscr
\usepackage[mathscr]{eucal}

\newcommand{\cprime}{\/{\mathsurround=0pt$'$}}

\begin{document}
\title[The Pavlov-Mikhalev equation]{On recursion operators for symmetries\\
  of the Pavlov-Mikhalev equation}

\author{I.S.~Krasil'shchik} \address{Trapeznikov Institute of Control
  Sciences, 65 Profsoyuznaya street, Moscow 117997, Russia}
\email{josephkra@gmail.com}\thanks{The work was partially supported by
  the Russian Science Foundation Grant 21-71-20034.}

\begin{abstract}
  In geometry of nonlinear partial differential equations, recursion
  operators that act on symmetries of an equation~$\mathcal{E}$ are understood
  as B\"{a}cklund auto-transformations of the equation~$\mathcal{TE}$ tangent
  to~$\mathcal{E}$. We apply this approach to a natural two-component
  extension of the 3D Pavlov-Mikhalev equation
  \begin{equation*}
    u_{yy} = u_{tx} + u_yu_{xx} - u_xu_{xy}.
  \end{equation*}
  We describe the Lie algebra of symmetries for this extension, construct two
  recursion operators (one of them was known earlier) and find their
  action. We also establish the hereditary property of these operators as well
  as their compatibility (in the sense of the Fr\"{o}licher-Nijenhuis
  bracket). We find also twelve additional operators which are degenerate in a
  sense (we call them \emph{queer}) and discuss their properties. In the
  concluding part, a geometrical background of two-component conservation laws
  for multi-dimensional equations is exposed together with its relations to
  differential coverings.
\end{abstract}



\maketitle

\tableofcontents

\section*{Introduction}
\label{sec:introduction}

In a recent paper~\cite{Kr-Ver-AMP}, we studied a two-component Lagrangian
extension of the so-called 3D rdDym equation and showed how the geometric
approach to nonlinear PDEs (see~\cite{AMS}) facilitates efficient construction
of recursion operators. Below we apply similar techniques to the the
Pavlov-Mikhalev equation
\begin{equation}
  \label{eq:1}
  u_{yy} = u_{tx} + u_yu_{xx} - u_xu_{xy},
\end{equation}
see~\cite{Mikh, Pav}, as well as~\cite{Dun-2004}. Equation~\eqref{eq:1}
belongs to the same class as the rdDym one: it is linearly
degenerate~\cite{Fer-Mos} and admits a Lax pair with non-removable
parameter~\cite{Dun-2004, Mikh, Pav}. A recursion operator for symmetries
of~\eqref{eq:1} was found by O.~Morozov and is described
in~\cite{Baran-etal-compar-stud}.

To be more precise, we deal with the system~$\mathcal{E}$
\begin{equation}
  \label{eq:2}
  \begin{array}{l}
    v_{yy} = v_{xt} - u_xv_{xy} - 2v_yu_{xx} + 2v_xu_{xy} + u_yv_{xx},\\[3pt]
    u_{yy} = u_{xt} - u_xu_{xy} + u_yu_{xx},
  \end{array}
\end{equation}
where the first equation is obtained by applying the adjoint linearization
of~\eqref{eq:1} to a new variable~$v$. The system may be considered as a
\emph{Lagrangian deformation} of the initial equation,
see~\cite{Baran-etal-Lagr-deform}, and coincides with cotangent equation of
the latter (see details in~\cite{K-V-Gdeq, Kra-Ver-Vit-Springer}).

The algorithm of construction and analysis of recursion operators that we use
below is based on the interpretation of these operators as B\"{a}cklund
auto-trans\-forma\-tions of the \emph{tangent equation},
see~\cite{Mar-another}, and essentially uses the theory of differential
coverings~\cite{VinKrasTrends}. It should be noted that the tangent equation
as an efficient instrument the theory of recursion operators was intensively
used in~\cite{Ker-Kras-Kluwer}.

The algorithm itself can be shortly described as follows\footnote{The exact
  meaning and details will be explained in
  Section~\ref{sec:basic-constr-notat}}:
\begin{itemize}
\item Consider the infinite prolongation of the equation under study as a
  submanifold~$\mathcal{E}\subset J^\infty(\pi)$ in an appropriate jet space.
\item Choose convenient internal coordinates in~$\mathcal{E}$ for particular
  computations, cf.~\cite{Mar-ort}.
\item Compute symmetries of~$\mathcal{E}$ solving the
  equation~$\ell_{\mathcal{E}}(\phi) = 0$, where~$\ell_{\mathcal{E}}$ is the
  linearization operator of~$\mathcal{E}$.
\item Construct the tangent equation~$\mathcal{TE}$ by adding the
  equations~$\ell_{\mathcal{E}}(q) = 0$ to the initial ones. Here~$q$ is a new
  \emph{odd} (and this is essential) dependent variable.
\item Find differential coverings over~$\mathcal{TE}$ linear in~$q$ and its
  derivatives. Usually (but not necessary) these coverings are associated with
  two-component conservation laws of~$\mathcal{TE}$. Let~$w$ denote the
  corresponding nonlocal variables.
\item Find nonlocal shadows of symmetries linear both in~$q$ and~$w$. A pair
  (covering, shadow) provides a B\"{a}cklund auto-transformation
  of~$\mathcal{TE}$, i.e., a recursion operator.
\item Study the action of the operator on symmetries.
\item For hereditary and compatibility properties, find nonlocal symmetries
  that correspond to the shadows, if the former exist. Their super-commutators
  are exactly the needed  Fr\"{o}licher-Nijenhuis brackets.
\end{itemize}
All these steps are accomplished in
Sections~\ref{sec:equat-its-symm}--\ref{sec:recurs-oper-their}. Section~\ref{sec:basic-constr-notat}
contains the theoretical background which is necessary for the subsequent
exposition. It should be noted that our results comprise two types of
recursion operators: two of them are quite ``conventional'', while twelve are
degenerate, ``queer''. Nevertheless, existence of these operators reflect
specific properties of the equation under study. We discuss them in
Section~\ref{sec:disc-queer-oper}.

Some of our computational results are rather voluminous and to make the reading
more comfortable we place them in
Appendices~\ref{sec:expl-form-symm}--\ref{sec:conv-pres}.

\section{Basic constructions and notation}
\label{sec:basic-constr-notat}

In our exposition here we follow the books~\cite{AMS, Kra-Ver-Vit-Springer}
and the paper~\cite{VinKrasTrends}. By a number of reasons, we stick to the
coordinate version.

\subsection{Jets}
\label{sec:jets}

Let~$\pi\colon \mathbb{R}^m\times\mathbb{R}^n = E \to \mathbb{R}^n = M$ be the
trivial bundle with the coordinates $x = (x^1,\dots,x^n)$ in~$\mathbb{R}^n$
(independent variables) and $u = (u^1,\dots,u^m)$ in~$\mathbb{R}^m$ (``unknown
functions''). The space of $k$-jets $J^k(\pi)$ is, in addition to~$x$ and~$u$,
endowed with the coordinates~$u_\sigma^j$, where~$\sigma$ is a symmetric
multi-index of length~$\leq k$. Variables~$u_\sigma^j$ correspond to the
partial derivatives~$\partial u^j/\partial x^\sigma$. Natural projections
\begin{equation*}
  \pi_k\colon J^k(\pi)\to M,\qquad \pi_{k,l}\colon J^k(\pi)\to J^l(\pi), \quad
  k\geq l,
\end{equation*}
are defined, and their inverse limit~$J^\infty(\pi)$ is called the space of
infinite jets. These projections define also the bundles
\begin{equation*}
  \pi_\infty\colon J^\infty(\pi)\to M,\qquad
  \pi_{\infty, k}\colon J^\infty(\pi)\to J^k(\pi)
\end{equation*}
and the embeddings
\begin{equation*}
  \mathcal{F}_l(\pi)\subset\mathcal{F}_k(\pi),\qquad
  \Lambda_l^i(\pi)\subset\Lambda_k^i(\phi),
\end{equation*}
where~$\mathcal{F}_l(\pi) = C^\infty(J^l(\pi))$,
$\Lambda_l^i(\pi) = \Lambda^i(J^l(\pi))$ are the $\mathbb{R}$-algebra of
smooth functions and the $\mathcal{F}_l(\pi)$-module of differential $i$-forms
on~$J^l(\pi)$, respectively. The corresponding objects on~$J^\infty(\pi)$ are
defined by
\begin{equation*}
  \mathcal{F}(\pi) = \bigcup_{l\geq0}\mathcal{F}_l(\pi),\qquad
  \Lambda^i(\pi) = \bigcup_{l\geq0}\Lambda_l^i(\pi).
\end{equation*}
A vector field on~$J^\infty(\pi)$ is by definition an $\mathbb{R}$-linear
derivation $X\colon \mathcal{F}(\pi)\to \mathcal{F}(\pi)$. The module of
vector fields is denoted by~$D(\pi)$.

Let $f\in \Gamma(\pi)$ be a (local) section of~$\pi$. Then the
section~$j_\infty(f) \in\Gamma(\pi_\infty)$ defined by
\begin{equation*}
  u_\sigma^j = \frac{\partial^\sigma f^j}{\partial x^\sigma}, \qquad
  j = 1,\dots, m,\quad \abs{\sigma}\geq 0,
\end{equation*}
is called the infinite jet of~$f$. Graphs of infinite jets passing through a
given point~$\theta\in J^\infty(\pi)$ are tangent to each other and their
common $n$-dimensional $\pi_\infty$-horizontal tangent plane is denoted
by~$\mathcal{C}_\theta$. The distribution $\mathcal{C}\colon \theta \mapsto
\mathcal{C}_\theta$ is called the Cartan distribution. It is spanned by the
vector fields
\begin{equation*}
  D_{x^i} = \frac{\partial}{\partial x^i} + \sum_{j,\sigma} u_{\sigma i}^j
  \frac{\partial}{\partial u_\sigma^j} \in D(\pi),
\end{equation*}
or dually, annihilates $1$-forms
\begin{equation*}
  \omega_\sigma^j = du_\sigma^j - \sum_i u_{\sigma i}^j\,dx^i\in \Lambda^1(\pi).
\end{equation*}
The field~$D_{x^i}$ is called the total derivative with respect to~$x^i$,
while~$\omega_\sigma^j$ are Cartan (or vertical) forms on~$J^\infty(\pi)$. The
Cartan distribution is Frobenius integrable, i.e., $[X,Y]\in \mathcal{C}$
whenever $X$, $Y\in \mathcal{C}$.

Finally, one has two natural splittings
\begin{equation*}
  D(\pi) = \mathcal{C} \oplus D^v(\pi), \qquad
  \Lambda^1(\pi) = \Lambda_v^1(\pi) \oplus \Lambda_h^1(\pi),
\end{equation*}
where~$D^v(\pi)$ is the submodule of $\pi_\infty$-vertical fields and
\begin{equation*}
  \Lambda_h^1(\pi) = \sd{\sum_i a_i\,dx^i}{a_i\in\mathcal{F}(\pi)},\quad
  \Lambda_v^1(\pi) = \sd{\sum_{j,\sigma}
    b_\sigma^j\omega_\sigma^j}{b_\sigma^j\in\mathcal{F}(\pi)} 
\end{equation*}
consist of horizontal and Cartan forms, respectively. Consequently,
\begin{equation*}
  \Lambda^i(\pi) = \bigoplus_{p+q=i}\Lambda_h^q(\pi)\wedge\Lambda_v^p(\pi),
\end{equation*}
where
\begin{equation}\label{eq:3}
  \Lambda_h^q(\pi) = \underbrace{\Lambda_h^1(\pi)\wedge \dots
    \wedge\Lambda_h^1(\pi)}_{q\text{ times}},\qquad
  \Lambda_v^p(\pi) = \underbrace{\Lambda_v^1(\pi)\wedge \dots
    \wedge\Lambda_v^1(\pi)}_{q\text{ times}}
\end{equation}
with the corresponding splitting of the de~Rham differential in the horizontal
and vertical (Cartan) parts:
\begin{equation}\label{eq:4}
  d_h = \sum_idx^i\wedge D_{x^i}, \qquad
  d_v = \sum_{j,\sigma} \omega_\sigma^j\wedge\frac{\partial}{\partial u_\sigma^j}.
\end{equation}

\subsection{Equations}\label{sec:equations}
Consider an $\mathcal{F}(\pi)$-module~$P$ of rank $r$ and its element
$F = (F^1,\dots,F^r)$. An (infinitely prolonged) partial differential equation
associated with~$F$ is
\begin{equation}\label{eq:10}
  \mathcal{E} = \sd{\theta\in J^\infty(\pi)}{\eval{D_\sigma(F^j)}_\theta = 0,
    \quad j=1,\dots,r,\ \abs{\sigma}\geq0},
\end{equation}
where~$D_\sigma$ denotes the composition of total derivatives corresponding to
the multi-index~$\sigma$. The restriction $\eval{\pi_\infty}_{\mathcal{E}}$
will be also denoted by~$\pi_\infty$. Solutions of~$\mathcal{E}$ are sections
of~$\pi$ such that the graphs of their infinite jets lie
in~$\mathcal{E}$. Functions and forms on~$\mathcal{E}$ are by definition
\begin{equation*}
  \mathcal{F}(\mathcal{E}) = \eval{\mathcal{F}(\pi)}_{\mathcal{E}}, \qquad
  \Lambda^i(\mathcal{E}) = \eval{\Lambda^i(\pi)}_{\mathcal{E}},
\end{equation*}
respectively. Vector fields on~$\mathcal{E}$ are derivations
$X\colon \mathcal{F}(\mathcal{E}) \to \mathcal{F}(\mathcal{E})$; the
$\mathcal{F}(\mathcal{E})$-module of these fields is denoted
by~$D(\mathcal{E})$.

The Cartan plane at~$\theta\in\mathcal{E}$ is defined as
$\mathcal{C}_\theta\cap T_\theta\mathcal{E}\subset T_\theta J^\infty(\pi)$. In
this way, we obtain an integrable distribution almost everywhere
on~$\mathcal{E}$. We also have splittings similar to~\eqref{eq:3}, as well as
the vertical and Cartan differentials defined exactly like in~\eqref{eq:4}.

It follows from~\eqref{eq:10} that the total derivatives can be restricted
to~$\mathcal{E}$. These restrictions will be also denoted by~$D_{x^i}$ and we
shall always assume that the only solutions of the system $D_{x^i}(f) = 0$, $i
= 1,\dots,n$, are constants. Such equations are called differentially
connected.

\subsection{Symmetries}
\label{sec:symmetries}

Everywhere below the word ``symmetry'' means an infinitesimal symmetry.

A symmetry of the Cartan distribution on~$J^\infty(\pi)$ is a a vector field
$X\in D^v(\pi)$ such that $[X, \mathcal{C}]\subset \mathcal{C}$. The set of
symmetries is a Lie algebra over~$\mathbb{R}$ denoted by~$\sym(\pi)$.
\begin{theorem}\label{thm:symmetries-1}
  There is a one-to-one correspondence between $\sym(\pi)$ and the module
  $\kappa = \Gamma(\pi_\infty^*(\pi))$\textup{,} where~$\pi_\infty^*(\pi)$
  denotes the pull-back. This correspondence is given by the formula
  \begin{equation*}
    \kappa\ni \phi\mapsto \Ev_\phi = \sum_{j,\sigma}
    D_\sigma(\phi^j)\frac{\partial}{\partial u_\sigma^j}\in D^v(\pi).
  \end{equation*}
\end{theorem}
The field~$\Ev_\phi$ is said to be evolutionary; $\phi$ being generating
section. In what follows, we do not distinguish between evolutionary fields
and their generating sections when possible.

Since $\sym(\pi)$ is closed with respect to the commutator,
Theorem~\ref{thm:symmetries-1} implies that
\begin{equation*}
  [\Ev_\phi, \Ev_{\phi'}] = \Ev_{\{\phi, \phi'\}}
\end{equation*}
for some element $\{\phi, \phi'\}\in\kappa(\pi)$ which is called the Jacobi
bracket of~$\phi$ and~$\phi'$. To describe it explicitly, recall the following
fact:

\begin{proposition}\label{prop:symmetries-2}
  Let $\xi$ be a vector bundle over~$M$. Then the component-wise action
  \begin{equation*}
    \Ev_\phi^\xi\colon \Gamma(\pi_\infty^*(\xi))\to \Gamma(\pi_\infty^*(\xi))
  \end{equation*}
  is well defined and
  \begin{equation*}
    \Ev_\phi^\xi(a F) = \Ev_\phi(a)F + a\Ev_\phi^\xi(F) 
  \end{equation*}
  for all $a\in \mathcal{F}(\pi)$ and $F\in \Gamma(\pi_\infty^*(\xi))$. Then
  \begin{equation}  \label{eq:5}
    \{\phi, \phi'\} = \Ev_\phi^\pi(\phi') - \Ev_{\phi'}^\pi(\phi).
  \end{equation}
\end{proposition}

Let now $\mathcal{E}\subset J^\infty(\pi)$ be an equation. Its symmetry is a symmetry of
the Cartan distribution on~$\mathcal{E}$. To describe the algebra
$\sym(\mathcal{E})$, let us, using Proposition~\ref{prop:symmetries-2}, define
the operator
\begin{equation*}
  \ell_F\colon \xi\to \Gamma(\pi_\infty^*(\xi)),\quad \ell_F(\phi) =
  \Ev_\phi^\xi(F).
\end{equation*}
In coordinates, one has
\begin{equation}
  \label{eq:6}
  \ell_F = 
  \begin{pmatrix}
    \sum_\sigma \dfrac{\partial F^\alpha}{\partial u_\sigma^\beta} D_\sigma
  \end{pmatrix}
\end{equation}

\begin{remark}
  An element $F\in \Gamma(\pi_\infty^*(\xi))$ is a nonlinear differential
  operator acting from sections of~$\pi$ to those of~$\xi$. Hence,~$\ell_F$ is
  its linearization. Note that~$\ell_F$ is an operator in total
  derivatives. Such operators (we call them $\mathcal{C}$-differential) admit
  restrictions to graphs of infinite jets and infinite prolongations.
\end{remark}

\begin{theorem}\label{thm:symmetries-3}
  Let $\mathcal{E}\subset J^\infty(\pi)$ be an equation associated with an
  element $F\in P = \Gamma(\pi_\infty^*(\xi))$ and such that
  $\pi_{\infty,0}(\mathcal{E}) = J^0(\pi)$. Then
  \begin{equation*}
    \sym(\mathcal{E}) = \ker\ell_{\mathcal{E}},
  \end{equation*}
  where $\ell_{\mathcal{E}}$ is the restriction of~$\ell_F$ to~$\mathcal{E}$.
\end{theorem}

\subsection{Conservation laws and cosymmetries}
\label{sec:cons-laws-cosymm}

A conservation law of~$\mathcal{E}$ is a $d_h$-closed differential form
$\omega \in \Lambda_h^{n-1}(\mathcal{E})$. It is trivial if $\omega =
d_h(\rho)$ for some $\rho\in \Lambda_h^{n-2}(\mathcal{E})$. Two conservation
laws are equivalent if they differ by a trivial one. In coordinates, if
\begin{equation*}
  \omega =
  \sum_{i=1}^na_i\,x^1\wedge\dots\wedge\,dx^{i-1}\wedge\,dx^{i+1}\wedge
  \dots\wedge\,dx^n,
\end{equation*}
then $\omega$ is a conservation law if and only if
\begin{equation}
  \label{eq:7}
  \sum_{i=1}^n (-1)^{i+1}D_{x^i}(a_i) =0.
\end{equation}
Direct computation of conservation laws using formula~\eqref{eq:7} is not
simple and is complicated by the existence of trivial laws, which are of no
interest. For the ``majority'' of equations, the procedure can be facilitated.

Let $\mathcal{E}$ be given by some~$F\in P$ and consider
$G\in Q = \Gamma(\pi_\infty^*(\eta))$, where~$\eta$ is another vector bundle
over~$M$. Assume~$\mathcal{E}$ to enjoy the following regularity property:
\begin{equation*}
  \eval{G}_{\mathcal{E}} = 0 \text{ implies } G = \Delta(F)
\end{equation*}
for some $\mathcal{C}$-differential operator $\Delta = \left(
\sum_\sigma a_{\beta,\sigma}^\alpha D_\sigma\right)\colon P\to Q$. For any such an
operator define its adjoint $\Delta^*\colon \hat{Q}\to \hat{P}$, where
$\hat{\bullet} = \hom_{\mathcal{F}(\mathcal{E})}(\bullet,
\Lambda_h^n(\mathcal{E}))$, by
\begin{equation*}
  \Delta^* = \left((-1)^{\abs{\sigma}} D_\sigma\circ
    a_{\alpha,\sigma}^\beta\right) 
\end{equation*}
and consider a form $\bar{\omega}\in \Lambda_h^n(\pi)$ such that
$\eval{\bar{\omega}}_{\mathcal{E}} = \omega$ for a conservation law~$\omega$
of the equation~$\mathcal{E}$. Then
\begin{equation}\label{eq:27}
  d_h\bar{\omega} = \Delta(F),\qquad \Delta\colon P\to\Lambda_h^n(\pi).
\end{equation}
Set $\psi_\omega = \eval{\Delta^*(1)}_{\mathcal{E}}$.
\begin{theorem}
  Let $\mathcal{E}$ be a $2$-line equation in the sense
  of~\cite{AMV-SecCal}. Then\textup{:}
  \begin{enumerate}
  \item The element $\psi_\omega\in\hat{P}$ is well defined\textup{,}
    i.e.\textup{,} does not depend on the choice of~$\bar{\omega}$.
  \item Conservation laws $\omega$ and $\omega'$ are equivalent if and only if
    $\psi_\omega = \psi_{\omega'}$. In particular\textup{,} $\omega$ is
    trivial if and only if~$\psi_\omega = 0$.
  \item The element $\psi_\omega$ enjoys the equation
    \begin{equation}
      \label{eq:8}
      \ell_{\mathcal{E}}^*(\psi_\omega) = 0,
    \end{equation}
    where $\ell_{\mathcal{E}}^*\colon \hat{P}\to \hat{\kappa}$.
  \end{enumerate}
\end{theorem}
We say that $\psi_\omega$ is the generating section of~$\omega$, while
elements $\psi\in\hat{P}$ satisfying Equation~\eqref{eq:8} are called
cosymmetries of~$\mathcal{E}$.

\subsection{Differential coverings}
\label{sec:diff-cover}

Let a smooth fiber bundle $\tau\colon \tilde{\mathcal{E}}\to \mathcal{E}$ be
such that the restriction
$\eval{d\tau}_{\tilde{\mathcal{C}}_{\tilde{\theta}}}$,
$\tilde{\theta}\in \tilde{\mathcal{E}}$, is an isomorphism
to~$\mathcal{C}_{\tau(\tilde{\theta})}$, where
$\tilde{\mathcal{C}}_{\tilde{\theta}}$
and~$\mathcal{C}_{\tau(\tilde{\theta})}$ are the Cartan planes at the
corresponding points. Such maps are called (differential) coverings.

Let $w^1,\dots, w^\alpha,\dots$ be coordinates in fibers of~$\tau$ in the
vicinity of some point~$\theta\in \mathcal{E}$ (nonlocal variables). Then the
condition that~$\tau$ is a covering means that the total derivatives
on~$\tilde{\mathcal{E}}$ are of the form $\tilde{D}_{x^i} = D_{x^i} + X_i$,
where
\begin{equation*}
  X_i = X_i^1\frac{\partial}{\partial w^1} + \dots +
  X_i^\alpha\frac{\partial}{\partial 
    w^\alpha} + \dots 
\end{equation*}
are $\tau$-vertical fields such that
\begin{equation}
  \label{eq:11}
  D_{x^i}(X_j) - D_{x^j}(X_i) + [X_i, X_j] = 0,\qquad i<j.
\end{equation}
This also means that~$\tilde{\mathcal{E}}$ may be understood as the
overdetermined system
\begin{equation}
  \label{eq:12}
  \frac{\partial w^\alpha}{\partial x^i} = X_i^\alpha,
\end{equation}
whose compatibility conditions coincide with~$\mathcal{E}$.

Two coverings $\tau_i\colon \tilde{\mathcal{E}}_i\to \mathcal{E}$, $i=1, 2$,
are equivalent if there exists an isomorphism of bundles
$f\colon \tilde{\mathcal{E}}_1 \to \tilde{\mathcal{E}}_2$ such that
$df(\mathcal{C}_\theta) = \mathcal{C}_{f(\theta)}$,
$\theta\in\tilde{\mathcal{E}}_1$. A covering is irreducible if the covering
equation is differentially connected and is trivial if it is equivalent to
$\tau_0\colon \mathcal{E}_0\times \mathcal{E}\to \mathcal{E}$ with
$\tilde{D}_{x^i} = D_{x^i}$. For any two coverings over~$\mathcal{E}$, the
Whitney product
$\tau_1\times \tau_2\colon \tilde{\mathcal{E}}_1\times_{\mathcal{E}}
\tilde{\mathcal{E}}_2\to \mathcal{E}$ is naturally endowed with the structure
of a covering and it can be shown that any covering is locally equivalent to
the Whitney product of a trivial and irreducible ones.

A covering $\tau$ given, we treat all the objects on~$\tilde{\mathcal{E}}$ to
be nonlocal with respect to~$\mathcal{E}$. In particular, a
symmetry~$\tilde{X}$ of~$\tilde{\mathcal{E}}$ is a nonlocal symmetry
of~$\mathcal{E}$. The defining equation for nonlocal symmetries is
$\ell_{\tilde{\mathcal{E}}}(\tilde{\phi}) = 0$. On the other hand, the
restriction
\begin{equation}\label{eq:13}
  X = \eval{\tilde{X}}_{\mathcal{F}(\mathcal{E})}\colon
  \mathcal{F}(\mathcal{E})\to \mathcal{F}(\tilde{\mathcal{E}})
\end{equation}
is an $\mathcal{F}(\tilde{\mathcal{E}})$-valued derivation that preserves the
Cartan distributions. Derivations of such a type are called shadows. We say
that a shadow lifts to~$\tilde{\mathcal{E}}$ if there exists an~$\tilde{X}$
such that~\eqref{eq:13} fulfils. Any shadow is an evolutionary derivation
on~$\mathcal{E}$ taking values in~$\mathcal{F}(\tilde{\mathcal{E}})$ with the
generating section~$\tilde{\phi}$ living on~$\tilde{\mathcal{E}}$ and
satisfying the defining equation
\begin{equation}
  \label{eq:14}
  \tilde{\ell}_{\mathcal{E}}({\tilde{\phi}}) = 0.
\end{equation}

\begin{remark}
  The notation~$\tilde{\ell}$ in~\eqref{eq:14} means the lift
  to~$\tilde{\mathcal{E}}$. Such a lift is possible for any
  $\mathcal{C}$-differential operator just by changing~$D_{x^i}$
  to~$\tilde{D}_{x^i}$.
\end{remark}

\subsection{B\"{a}cklund transformations}
\label{sec:backl-transf}

A B\"{a}cklund transformation $\mathcal{B}(\mathcal{E}_1, \mathcal{E}_2)$
between equations~$\mathcal{E}_1$ and~$\mathcal{E}_2$ is a diagram of the form
\begin{equation*} \xymatrix{
&\tilde{\mathcal{E}}\ar[dl]_-{\tau_1}\ar[dr]^-{\tau_2}&\\
\mathcal{E}_1&&\mathcal{E}_2\rlap{,} }
\end{equation*} where~$\tau_i$ are coverings. If the equation~$\mathcal{E}_i$,
$i = 1,2$, is imposed on unknowns~$u_i$ then $\mathcal{E}$ is an equation both
on~$u_1$ and~$u_2$ posesing the following characteristic property: if $(u_1,
u_2)$ solves~$\mathcal{E}$ and~$u_1$ solves~$\mathcal{E}_1$, then~$u_2$
solves~$\mathcal{E}_2$ and vice versa. When~$\mathcal{E}_1 = \mathcal{E}_2$,
$\mathcal{B}$ is called an auto-transformation.

Consider B\"{a}cklund transformations $\mathcal{B}_{12}(\mathcal{E}_1,
\mathcal{E}_2)$ and~$\mathcal{B}_{23}(\mathcal{E}_2, \mathcal{E}_3)$. Then the
diagram
\begin{equation*} \xymatrix{
&&\mathcal{E}_{12}\times_{\mathcal{E}_{2}}\ar[dl]_-{\tau_2^*(\tau_3)}
\ar[dr]^-{\tau_3^*(\tau_2)}\mathcal{E}_{23}&&\\
&\mathcal{E}_{12}\ar[dl]_-{\tau_1}\ar[dr]^-{\tau_2}
&&\mathcal{E}_{23}\ar[dl]_-{\tau_3}\ar[dr]^-{\tau_4}&\\
\mathcal{E}_{1}&&\mathcal{E}_{2}&&\mathcal{E}_{3} }
\end{equation*} provides a B\"{a}cklund transformation between~$\mathcal{E}_1$
and~$\mathcal{E}_3$, which is called the composition of~$\mathcal{B}_{12}$
and~$\mathcal{B}_{23}$.
\begin{remark} It may happen that the top equation is not differentially
closed. Then one should restrict the considerations onto irreducible leaves of
the Cartan distribution.
\end{remark}

\subsection{The tangent covering and recursion operators}
\label{sec:tangent-covering} Consider an equation $\mathcal{E}$ and its
tangent bundle $T\mathcal{E}\to \mathcal{E}$. Take the quotient bundle
\begin{equation}
  \label{eq:15} \mathbf{t}\colon \mathcal{TE} = T\mathcal{E}/\mathcal{C} \to
\mathcal{E}
\end{equation} and assume that its fibers are odd. Then~\eqref{eq:15} is
called the tangent covering to~$\mathcal{E}$ and~$\mathcal{TE}$ is called the
tangent equation. Locally, sections of~$\mathbf{t}$ may be understood as
$\pi_\infty$-vertical vector fields on~$\mathcal{E}$.
\begin{theorem}\label{thm:tangent-covering-1} The tangent covering posseses
the following properties\textup{:}
  \begin{enumerate}
  \item\label{thm:4-item:1} Sections of~$\mathbf{t}$ that preseve the Cartan
distributions are identified with symmetries of~$\mathcal{E}$.
  \item\label{thm:4-item:2} The superalgebra of functions on~$\mathcal{TE}$ is
canonically isomorphic to the Grassmann algebra~$\Lambda_v^*(\mathcal{E})$.
  \item If $\mathcal{E}$ is given by $F(u) = 0$\textup{,} then $\mathcal{TE}$
is given by the system $\{F(u) = 0, \ell_F(q) = 0\}$.
  \item\label{item:4-1} The Cartan differential defines an odd nilpotent
vector field~$\dc$ on~$\mathcal{TE}$\textup{,} such that $\dc(u_\sigma^j)
=q_\sigma^j$ and~$\dc(q_\sigma^j) =0$.
  \end{enumerate}
\end{theorem}

Statement~(\ref{thm:4-item:2}) of Theorem~\ref{thm:tangent-covering-1} implies
that fiber-wise linear functions on~$\mathcal{TE}$ are identified with
$\mathcal{C}$-differential operators $\kappa \to \mathcal{F}(\mathcal{E})$:
\begin{equation*} \phi\mapsto i_{\Ev_\phi}(\Upsilon),\qquad \Upsilon= \sum
a_\sigma^jq_\sigma^j,
\end{equation*} where $i$ denotes the inner product. It also follows from
Statement~(\ref{thm:4-item:1}) that B\"{a}cklund auto-transformations
of~$\mathcal{TE}$ relate symmetries of~$\mathcal{E}$ to each other, i.e., are
interpreted as recursion operators. Below, we construct these operators as
follows.

Choose internal coordinated~$x^i$, $u_\sigma^j$ in~$\mathcal{E}$ and
let~$q_\sigma^j$ be the corresponding coordinates in fibers of~$\mathbf{t}$.
Let also
\begin{multline*} \omega^\alpha = (X_{i_1}^\alpha\,dx^{i_1} +
X_{i_2}^\alpha\,dx^{i_2})\wedge\,dx^1\wedge \dots\wedge
\,dx^{i_1-1}\wedge\,dx^{i_1+1}\wedge\dots\\ \dots\wedge
\,dx^{i_2-1}\wedge\,dx^{i_2+1}\wedge\dots\wedge\,dx^n,\qquad \alpha =
1,\dots,s,
\end{multline*} be two-component conservation laws of~$\mathcal{TE}$ linear in
the variables~$q_\rho^j$. Then the system of relations
\begin{equation}\label{eq:16}
  \begin{array}{lcll} w_{\rho, x^l}^\alpha
    &=&
        D_\rho(X_l^\alpha),&\quad l = i_1,i_2,\\
    w_{\rho,x^l}^\alpha&= &w_{\rho l}^\alpha,&\quad l\neq i_1, i_2,
  \end{array}
\end{equation} defines a covering~$\tau^{(\omega)}\colon W^{(\omega)}\to
\mathcal{TE}$ over~$\mathcal{TE}$. Here~$\rho$ is a symmetric multi-index that
does not contain~$i_1$, $i_2$ and the nonlocal variables~$w_\rho^\alpha$ are
odd; in the case~$n = 2$ the second group of relations is void.

Let now
\begin{equation*}
  \phi^j = \sum_{\sigma,\beta} a_{\sigma,\beta}^j
  q_\sigma^\beta + \sum_{\rho,\alpha} b_{\rho,\alpha}^j w_\rho^\alpha, \quad j =
  1, \dots,m,
\end{equation*} be a shadow in $\tau^{(\omega)}$. Consider a second copy
of~$\mathcal{TE}$ with fiber-wise coordinates~$\bar{q}_\sigma^j$ and the map
$\tau_\phi\colon W^{(\omega)} \to \mathcal{TE}$ given by
\begin{equation}
  \label{eq:18} \bar{q}_\sigma^j = D_\sigma(\phi^i).
\end{equation} This map provides another covering $\tau^\phi\colon
W^{(\omega)}\to \mathcal{TE}$ and the resulting B\"{a}ck\-lund
transformation~$\mathcal{R}$, described by Equations~\eqref{eq:16}
and~\eqref{eq:18}, and is the desired recursion operator: substituting a known
symmetry instead of~$q$ and solving the system with respect to~$\bar{q}$, we
get the action.

\begin{remark} A ``tradition'' prescribes to eliminate nonlocal variables
from~\eqref{eq:16} and~\eqref{eq:18} and present a recursion operator as a
$\mathcal{C}$-differential relation between~$q$ and~$\bar{q}$. This may
convenient in simple cases, but in more complicated ones leads to practically
unreadable formulas (see Appendix~\ref{sec:conv-pres}).
\end{remark}

Let~$\mathcal{R}$ be a recursion operator and~$\phi$ be the corresponding
shadow. Denote by~$\bar{\phi}$ the lift of the latter to~$W^{(\omega)}$ if it
exists. Given another liftable shadow~$\phi'$, we have $[\Ev_{\bar{\phi}},
\Ev_{\bar{\phi}'}] = \Ev_{\ldb \bar{\phi}, \bar{\phi}'\rdb}$ for some nonlocal
symmetry~$\ldb \bar{\phi}, \bar{\phi}'\rdb$ of parity~$2$ (since
$\Ev_{\bar{\phi}}$ and~$\Ev_{\bar{\phi}'}$ are odd vector fields of
parity~$1$, their commutator is $\Ev_{\bar{\phi}}\circ\Ev_{\bar{\phi}'} +
\Ev_{\bar{\phi}'}\circ\Ev_{\bar{\phi}}$). We say that $\ldb \bar{\phi},
\bar{\phi}'\rdb$ is the Fr\"{o}licher-Nijenhuis bracket of~$\phi$
and~$\bar{\phi}$. A recursion operator is said to be hereditary if $\ldb
\bar{\phi}, \bar{\phi}\rdb = 0$; two operators are compatible if~$\ldb
\bar{\phi}, \bar{\phi}'\rdb = 0$.

\section{The equation and its symmetries}
\label{sec:equat-its-symm}

The two-component Pavlov-Mikhalev equation~$\mathcal{E}$ reads
\begin{equation}
  \label{eq:19}
  \begin{array}{rcl}
    u_{yy} &=& u_{xt} - u_x u_{xy} + u_yu_{xx},\\
    v_{yy} &=& v_{xt} - u_x v_{xy} - 2 v_y u_{xx} + 2 v_x u_{xy} + u_y v_{xx}.
  \end{array}
\end{equation}
We choose the functions
\begin{align*}
  &u_{0,k,l} = u_{\underbrace{\scriptstyle{x\dots x}}_{k\text{
  times}}\underbrace{\scriptstyle{t\dots 
  t}}_{\text{l times}}}, &&u_{1,k,l} = u_{y\underbrace{\scriptstyle{x\dots
  x}}_{k\text{ times}}\underbrace{\scriptstyle{t\dots t}}_{\text{l times}}},\\ 
  &v_{0,k,l} = v_{\underbrace{\scriptstyle{x\dots x}}_{k\text{
  times}}\underbrace{\scriptstyle{t\dots 
  t}}_{\text{l times}}}, &&v_{1,k,l} = v_{y\underbrace{\scriptstyle{x\dots
  x}}_{k\text{ times}}\underbrace{\scriptstyle{t\dots t}}_{\text{l times}}}
\end{align*}
for intenal coordinates on~$\mathcal{E}$. The total derivatives acquire the
form
\begin{align*}
  D_x&=\sum_{k,l\geq0}\left(u_{0,k+1,l}\frac{\partial}{u_{0,k,l}} +
       u_{1,k+1,l}\frac{\partial}{u_{1,k,l}} +
       v_{0,k+1,l}\frac{\partial}{v_{0,k,l}} +
       v_{1,k+1,l}\frac{\partial}{v_{1,k,l}}\right),\\ 
  D_t&=\sum_{k,l\geq0}\left(u_{0,k,l+1}\frac{\partial}{u_{0,k,l}} +
       u_{1,k,l+1}\frac{\partial}{u_{1,k,l}} +
       v_{0,k,l+1}\frac{\partial}{v_{0,k,l}} +
       v_{1,k,l+1}\frac{\partial}{v_{1,k,l}}\right),
  \intertext{and}
  D_y&=\sum_{k,l\geq0}\left(u_{1,k,l}\frac{\partial}{u_{0,k,l}} +
       D_x^kD_t^l(U)\frac{\partial}{u_{1,k,l}} +
       v_{1,k,l}\frac{\partial}{v_{0,k,l}} +
       D_x^kD_t^l(V)\frac{\partial}{v_{1,k,l}}\right)
\end{align*}
in these coordinates, where $U$ and $V$ are the right-hand sides
of~\eqref{eq:19}.

Equation~\eqref{eq:14} acquires the form
\begin{equation}\label{eq:20}
    \begin{array}{rcl}
    D_y^2(\phi^u)
    & =& D_xD_t(\phi^u) + u_{xx}D_y(\phi^u)  + u_yD_x^2(\phi^u) -
         u_xD_xD_y(\phi^u) \\
    &-& u_{xy}D_x(\phi^u),\\[4pt]
    D_y^2(\phi^v)
    &=&   2v_xD_xD_y(\phi^u) - v_{xy}D_x(\phi^u) + v_{xx}D_y(\phi^u) \\ 
    &-&2v_yD_x^2(\phi^u) + D_xD_t(\phi^v)
        + 2u_{xy}D_x(\phi^v) - 2u_{xx}D_y(\phi^v)  \\
    &+&  u_yD_x^2(\phi^v) -  u_xD_xD_y(\phi^v). 
  \end{array}
\end{equation}
Solving this system, one sees that its space of solutions $\sym(\mathcal{E})$
is generated over~$\mathbb{R}$ by the functions\footnote{Explicit presentation
  of the generators see in Appendix~\ref{sec:expl-form-symm}.}
\begin{equation*}
  \phi_1,\ \phi_2,\ \phi_3[\vartheta],\dots, \phi_6[\vartheta],\
  \phi_7,\dots,\phi_{14},\ 
  \phi_{15}[\vartheta], \dots,\phi_{18}[\vartheta],
\end{equation*}
where~$\phi_i = (\phi_i^u, \phi_i^v)$ and $\vartheta = \vartheta(t)$ is an
arbitrary smooth function in~$t$. The Lie algebra structure
of~$\sym(\mathcal{E})$ is described~in

\begin{proposition}
  \label{prop:equat-its-symm-1}
  Let
  \begin{align*}
    \mathfrak{s}(2)
    &= \langle\phi_1,\phi_2\rangle,\\
    \mathfrak{a}(1)
    &= \langle\phi_7\rangle,\\
    \mathfrak{a}(7)
    &= \langle\phi_8,\dots,\phi_{14}\rangle,\\
    \mathfrak{at}(4)
    &=
      \langle\phi_{15}[\vartheta],\dots,\phi_{18}[\vartheta]\rangle,\\ 
    \mathfrak{gt}(4)
    &=\langle\phi_3[\vartheta],\dots,\phi_6[\vartheta]\rangle.
  \end{align*}
  Then
  \begin{equation*}
    \sym(\mathcal{E}) = \big(\mathfrak{s}(2)\oplus\mathfrak{a}(1)\big) \ominus
    \mathfrak{gt}(4) \ominus \big(\mathfrak{at}(4)\oplus\mathfrak{a}(7)\big),
  \end{equation*}
  where $\ominus$ denotes a semidirect product. The first summand is the
  $2$-dimensional solvable Lie algebra with the commutator $\{\phi_2, \phi_1\}
  = \phi_1$\textup{,} the algebras~$\mathfrak{a}(1)$\textup{,}
  $\mathfrak{at}(4)$ and $\mathfrak{a}(7)$ are Abelian. The structure of
  $\mathfrak{gt}(4)$ is given by
  \begin{gather*}
    \{\phi_3[\vartheta], \phi_6[\bar{\vartheta}]\} = \phi_3[\vartheta_t
    \bar{\vartheta} - \vartheta\bar{\vartheta}_t],\quad \{\phi_4[\vartheta],
    \phi_5[\bar{\vartheta}]\} = \phi_3[\vartheta\bar{\vartheta}_t -
    \vartheta_1 \bar{\vartheta}],\\
    \{\phi_4[\vartheta], \phi_6[\bar{\vartheta}]\} = -\phi_4[\vartheta
    \bar{\vartheta}_t + \vartheta_t\bar{\vartheta}], \quad
    \{\phi_5[\vartheta], \phi_5[\bar{\vartheta}]\} =
    \phi_5[\bar{\vartheta}\vartheta_t - \bar{\vartheta}_t\vartheta],\\
    \{\phi_5[\vartheta], \phi_6[\bar{\vartheta}]\} =
    -\phi_5[\bar{\vartheta}\vartheta_t + \bar{\vartheta}_t\vartheta], \quad
    \{\phi_6(\vartheta), \phi_6[\bar{\vartheta}]\} =
    \phi_6[\bar{\vartheta}\vartheta_t - \bar{\vartheta}_t\vartheta]. 
  \end{gather*}
  The actions of $\mathfrak{s}(2)$ and $\mathfrak{a}(1)$ are
  \begin{gather*}
    \{\phi_1,\phi_4[\vartheta]\} = -2\phi_3[\vartheta], \quad
    \{\phi_1,\phi_5[\vartheta]\} = \phi_4[\vartheta], \quad \{\phi_1,\phi_9\}
    = -\phi_8,\\\
    \{\phi_1,\phi_{10}\} = -2\phi_9, \quad \{\phi_1,\phi_{11}\} = -3\phi_{10},
    \quad \{\phi_1,\phi_{12}\} = -4\phi_{11}, \\
    \{\phi_1,\phi_{13}\} = -5\phi_{12},\quad \{\phi_1,\phi_{14}\} =
    -6\phi_{13}, \quad \{\phi_1,\phi_{16}[\vartheta]\} =
    -2\phi_{15}[\vartheta],\\
    \{\phi_1,\phi_{17}[\vartheta]\} = -5\phi_{16}[\vartheta],\quad
    \{\phi_1,\phi_{18}[\vartheta]\} = -6\phi_{17}[\vartheta],\\[4pt]
    \{\phi_2,\phi_3[\vartheta]\} = 3\phi_3[\vartheta],\quad
    \{\phi_2,\phi_4[\vartheta]\} = 2\phi_4[\vartheta], \quad
    \{\phi_2,\phi_5[\vartheta]\} = \phi_5[\vartheta],\\
    \{\phi_2,\phi_9\} = 2\phi_9,,\quad \{\phi_2,\phi_{10}\} = \phi_{10}, \quad
    \{\phi_2,\phi_{12}\} = -\phi_{12},\\
    \{\phi_2,\phi_{13}\} = -2\phi_{13},\quad \{\phi_2,\phi_{14}\} =
    -3\phi_{14},  \quad\{\phi_2,\phi_{16}[\vartheta]\} = -\phi_{16}[\vartheta],\\
    \{\phi_2,\phi_{17}[\vartheta]\} = -2\phi_{17}[\vartheta],\quad
    \{\phi_2,\phi_{18}[\vartheta]\} = -3\phi_{18}[\vartheta]
  \end{gather*}
  and $\{\phi_7,\phi_i\} = - \phi_i$ for all $i =
  8,\dots,18$. Finally\textup{,} one has the following action
  of~$\mathfrak{gt}(4)$ on $\mathfrak{a}(7)\oplus \mathfrak{at}(4)$\textup{:}
  \begin{gather*}
    \{\phi_3[\vartheta],\phi_{14}\} = \phi_{15}[\vartheta_{ttt}], \quad
    \{\phi_3[\vartheta],\phi_{18}[\bar{\vartheta}]\} = \phi_{15}[\vartheta
    \bar{\vartheta}_t/2 + \vartheta_t\bar{\vartheta}],\\
     \{\phi_4[\vartheta],\phi_{13}\} = -\phi_{15}[\vartheta_{ttt}], \quad
     \{\phi_4[\vartheta],\phi_{14}\} = -2\phi_{16}[\vartheta_{ttt}], \\
     \{\phi_4[\vartheta],\phi_{17}[\bar{\vartheta}]\} =
     -\phi_{15}[\vartheta_{t} \bar{\vartheta} +
     \vartheta\bar{\vartheta}_t/2]\\
     \{\phi_5[\vartheta],\phi_{12}\} = -\phi_{15}[\vartheta_{ttt}], \quad
     \{\phi_5[\vartheta],\phi_{13}\} = -2\phi_{16}[\vartheta_{ttt}], \\
     \{\phi_5[\vartheta],\phi_{14}\} = -2\phi_{17}[\vartheta_{ttt}],\quad
     \{\phi_5[\vartheta],\phi_{16}[\bar{\vartheta}]\} = -\phi_{15}[\vartheta_t
     \bar{\vartheta} + \vartheta\bar{\vartheta}_t/2],\\
     \{\phi_5[\vartheta],\phi_{17}[\bar{\vartheta}]\} =
     -\phi_{17}[2\vartheta_t \bar{\vartheta} + \vartheta\bar{\vartheta}_t],
     \quad \{\phi_5[\vartheta],\phi_{18}[\bar{\vartheta}]\} =
     -\phi_{16}[2\vartheta_t \bar{\vartheta} + \vartheta\bar{\vartheta}_t],\\
      \{\phi_6[\vartheta], \phi_{11}\} = -\phi_{15}[\vartheta_{ttt}],\quad
      \{\phi_6[\vartheta], \phi_{12}\} = -2\phi_{16}[\vartheta_{ttt}],\\
      \{\phi_6[\vartheta], \phi_{13}\} = -2\phi_{17}[\vartheta_{ttt}], \quad
      \{\phi_6[\vartheta], \phi_{14}\} = -2\phi_{18}[\vartheta_{ttt}],\\
      \{\phi_6[\vartheta], \phi_{15}[\bar{\vartheta}]\} =
      -\phi_{15}[2\vartheta_t \bar{\vartheta} +
      \vartheta\bar{\vartheta}_t],\quad \{\phi_6[\vartheta],
      \phi_{16}[\bar{\vartheta}]\} = -\phi_{16}[2\vartheta_t \bar{\vartheta} +
      \vartheta\bar{\vartheta}_t],\\
      \{\phi_6[\vartheta], \phi_{17}[\bar{\vartheta}]\} =
      -\phi_{17}[2\vartheta_t \bar{\vartheta} +
      \vartheta\bar{\vartheta}_t],\quad \{\phi_6[\vartheta],
      \phi_{18}[\bar{\vartheta}]\} = -\phi_{18}[2\vartheta_t \bar{\vartheta} +
      \vartheta\bar{\vartheta}_t].
  \end{gather*}
  All the rest brackets vanish.
\end{proposition}

\section{The tangent equation: conservation laws and coverings}
\label{sec:cons-laws-cover}

The tangent covering is obtained by adding the equations
\begin{equation}\label{eq:28}
  \begin{array}{rcl}
  p_{yy} &=& u_yp_{xx} - u_{xy}p_x - u_xp_{xy} + u_{xx}p_y + p_{xt},\\
  q_{yy} &=& 2v_xp_{xy} + v_{xx}p_y - v_{xy}p_x - 2v_yp_{xx} + 2u_{xy}q_x \\
         &+&   u_yq_{xx} - u_xq_{xy} - 2u_{xx}q_y + q_{xt}
  \end{array}
\end{equation}
to the initial system, where $p$ and $q$ are odd coordinate functions in the
fibers of the covering. We were looking for two-component conservation laws of
second order in all jet variables\footnote{An attempt to raise the order makes
  computations unrealistically time-consuming.} and linear in~$p_\sigma$
and~$q_\sigma$. This led us to four conservation laws
$\omega_i = (X_i\,dx + Y_i\,dy)\wedge\, dt$, $i = 1,\dots,4$, with the
following components:
\begin{align*}
  X_1 &= 2u_xp_x  + p_y,\\
  Y_1 &= u_yp_x + u_xp_y + p_t,\\[4pt]
  X_2 &= q_y - v_xp_x - u_xq_x ,\\
  Y_2 &= q_t -2v_yp_x + v_xp_y + u_yq_x - 2u_xq_y,\\[4pt]
  X_3 &= \left(\frac{3}{2}u_x^2 - yu_{xt}\right) p_x + \left(yu_{xx} +
        \frac{1}{2}\right)p_t - u_{xy}p + u_xp_y - \frac{1}{2}yp_{yt},\\ 
  Y_3 &= (u_xu_{xy} - u_yu_{xx} - u_{xt})p + \left(u_xu_y -
        \frac{1}{2}yu_{yt}\right)p_x + \frac{1}{2}(u_x^2 - yu_{xt})p_y \\
      &+ (yu_{xy} + u_x)p_t - \frac{1}{2}y(u_yp_{xt} - u_xp_{yt} + p_{tt}),\\
  X_4&=  y(v_{xx}p_t  - v_{xt}p_x + u_{xx}q_t - u_{xt}q_x + q_{yt}) - v_{xy}p
       + v_xp_y - u_{xy}q + u_xq_y - q_t,\\ 
  Y_4 &= (u_xv_{xy} + 2v_yu_{xx} - 2v_xu_{xy} - u_yv_{xx} - v_{xt}) p
        + (u_y v_x - 2u_xv_y- 2yv_{yt})p_x \\
  &+ (u_xv_x + yv_{xt}) p_y + (yv_{xy} + v_x)p_t + (u_xu_{xy} - u_yu_{xx} -
    u_{xt})q \\
      &+ (u_xu_y + yu_{yt})q_x - (u_x^2 + 2yu_{xt})q_y + (yu_{xy} + u_x)q_t\\
  &+ y(2v_xp_{yt} - 2v_yp_{xt} + u_yq_{xt} - u_xq_{yt} + q_{tt}).
\end{align*}

\begin{remark}
  Actually, all the coefficients may be multiplied by an arbitrary
  function~$f(t)$, but this does not influence the final results.
\end{remark}

The generating sections of these conservation laws are
\begin{equation*}
  \psi_1 =
  \begin{pmatrix}
    0\\ 0\\ 1\\ 0
  \end{pmatrix},\quad \psi_2 =
  \begin{pmatrix}
    0\\ 0\\ 0\\ 1
  \end{pmatrix},\quad \psi_3 =
  \begin{pmatrix}
    p_x\\ 0\\ u_x\\ 0
  \end{pmatrix},\quad \psi_4 =
  \begin{pmatrix}
    q_x\\ p_x\\ v_x\\ u_x
  \end{pmatrix},
\end{equation*}
so they are all nontrivial.

We shall consider below the coverings
\begin{equation}\label{eq:21}
  w_{i, x} = X_i,\quad w_{i, y} = Y_i,\qquad i = 1,\dots,4,
\end{equation}
with the nonlocal variables~$w_i, w_{i,t}, w_{i,tt},\dots$ and their Whitney
product $\tau\colon W\to \mathcal{TE}$.

\section{Shadows and lifts}
\label{sec:shadows-lifts}

We now lift the linearization operator~$\ell_{\mathcal{E}}$ to the
covering~\eqref{eq:21} and solve the equation
$\tilde{\ell}_{\mathcal{E}}(\Phi) = 0$, where $\Phi = (\Phi^u, \Phi^v)$ is of
second jet order and linear in all the odd variables ($p_\sigma$, $q_\sigma$,
and~$w_{i, \sigma}$). Under these assumptions, the equation has three
solutions:
\begin{align}
  \label{eq:22}
  &\begin{array}{rcl}
    \Phi_0^u &=& p,\\
    \Phi_0^v &=& q;
  \end{array}\\
  &\begin{array}{rcl}\label{eq:23}
    \Phi_1^u &=& u_xp - w_1,\\
    \Phi_1^v &=& v_xp - 2u_xq - w_2;
  \end{array}\\\label{eq:24}
  &\begin{array}{rcl}
     \Phi_2^u &=& u_yp - 2yu_xp_t - u_xw_1 + yw_{1,t} + 2w_3,\\
     \Phi_2^v &=& -2v_y p + yv_xp_t + (3u_x^2 + u_y)q + yu_xq_t\\
     &&- v_xw_1 + 2u_xw_2 + yw_{2,t} - w_4. 
  \end{array}
\end{align}
The solution~$\Phi_0$ corresponds to the identical recursion operators and
thus is of no interest, while~$\Phi_1$ and~$\Phi_2$ will be studied in more
detail.

\begin{proposition}\label{prop:shadows-lifts-1}
  The shadows~$\Phi_1$ and~$\Phi_2$ can be lifted both to the tangent
  covering $\mathbf{t}\colon \mathcal{TE}\to \mathcal{E}$\textup{,} and to
  $\tau\colon W \to \mathcal{TE}$.
\end{proposition}

\begin{proof}
  The proof is accomplished in two technically different steps.

  \textbf{Step~1} consists in lifting to $\mathcal{TE}$ and is based on
  Statements~\ref{thm:4-item:2} and~\ref{item:4-1} of
  Theorem~\ref{thm:tangent-covering-1}. Namely, the equalities
  \begin{equation*}
    \Ev_\Phi(p_\sigma) = L_{\Ev_\Phi}(\dc(u_\sigma)) =
    -\dc(L_{\Ev_\Phi}(u_\sigma)) =  -\dc(D_\sigma(\Phi^u))
  \end{equation*}
  and similar for~$\Ev_\Phi(q_\sigma)$ are to be satisfied (since
  both~$\Ev_\Phi$ and~$\dc$ are odd fields, they anticommute). To compute the
  last term, one needs to find the action of~$\dc$ on~$w_i$. To this end, we
  solve the equations
  \begin{equation*}
    D_x(\dc(w_i)) = \dc(X_i),\quad  D_y(\dc(w_i)) = \dc(Y_i),\qquad i =
    1,\dots, 4,
  \end{equation*}
  and obtain $ \dc(w_1) = \dc(w_2) = 0$ and
  \begin{equation*}
   \dc(w_3) = yp_xp_t  + pp_y,\qquad \dc(w_4) = y(p_xq_t - p_tq_x) + pq_y -
   p_yq. 
 \end{equation*}
 Consequently,
 \begin{align*}
   &\begin{array}{rcl}
      \Phi_1^p &=& pp_x,\\
      \Phi_1^q &=& 2p_xq + pq_x;
    \end{array}\\
   &\begin{array}{rcl}     
      \Phi_2^p &=& p_x w_1 - p p_y,\\
      \Phi_2^q &=& -6u_xp_xq - 2p_xw_2 - 2p_yq + q_xw_1 - pq_y,
   \end{array}
 \end{align*}
 and this finishes the first step.

 \textbf{Step~2.} To accomplish the second step, we solve the equations
 \begin{equation*}
   D_x(\Ev_{\Phi_j}(w_i)) = \Ev_{\Phi_j}(X_i),\quad  D_y(\Ev_{\Phi_j}(w_i)) =
   \Ev_{\Phi_j}(Y_i),\qquad i = 
    1,\dots, 4,
  \end{equation*}
  for $j = 1$, $2$. The results are presented in
  Appendix~\ref{sec:expl-form-lifts}.
\end{proof}

\section{Recursion operators and their action}
\label{sec:recurs-oper-their}

Thus, we have two recursion operators $\mathcal{R}_1$ and~$\mathcal{R}_2$ that
correspond to the found shadows: the first is described by
Equations~\eqref{eq:23} combined with the defining equations of the coverings;
similarly, the second one is obtained from Equations~\eqref{eq:24}. But such a
presentation, as it was indicated above, does not comply with the existing
tradition. To obtain the conventional form, we get rid of nonlocal variables
and for the first operator obtain the system
\begin{align*}
  D_x(\tilde{\phi}^u) &= u_{xx}\phi^u - u_xD_x(\phi^u) - D_y(\phi^u),\\
  D_y(\tilde{\phi}^u) &= u_{xy}\phi^u - u_yD_x(\phi^u) - D_t(\phi^u),\\
  D_x(\tilde{\phi}^v) &= v_{xx}\phi^u + 2v_xD_x(\phi^u) - 2u_{xx}\phi^v -
                        u_xD_x(\phi^v) - D_y(\phi^v),\\ 
  D_y(\tilde{\phi}^v) &= v_{xy}\phi^u + 2v_yD_x(\phi^u) - 2u_{xy}\phi^v -
                        u_yD_x(\phi^v) - D_t(\phi^v). 
\end{align*}
Note that the first two equations provide the known recursion operator of the
one-component Pavlov-Mikhalev equation, see~\cite{Baran-etal-compar-stud}. A
similar form for the second operator is quite complicated and we present it in
Appendix~\ref{sec:conv-pres}.

\begin{remark}
  Oleg Morozov found a simpler presentation of this operator.
\end{remark}

\begin{proposition}
  The operators~$\mathcal{R}_1$ and~$\mathcal{R}_2$ are hereditary and
  compatible\textup{,} i.e.\textup{,}
  \begin{equation*}
    \ldb\Phi_1,\Phi_1\rdb = \ldb\Phi_2,\Phi_2\rdb = \ldb\Phi_1,\Phi_2\rdb = 0.
  \end{equation*}
\end{proposition}
 
\begin{proof}
  The result is proved by tiresome, but quite straightforward computations.
\end{proof}

Let us now describe the action of the constructed operators on symmetries
of~$\mathcal{E}$. Note first that the operator~$\mathcal{R}_1$ relates the
trivial symmetry~$\phi = 0$ with symmetries of the form
\begin{equation*}
  \phi_3[\vartheta_1] + \phi_{15}[\vartheta_2].
\end{equation*}
In a similar way, $\mathcal{R}_2$ transforms zero to
\begin{equation*}
  \phi_3[\vartheta_1] + \phi_{15}[\vartheta_2] + \phi_4[\vartheta_3] +
  \phi_{16}[\vartheta_4]. 
\end{equation*}
All the actions below will be presented modulo these sets of symmetries. Note
also that in a number of cases the result action is a nonlocal shadow; the
latter will be denoted by~$\nu_i$ (their exact form is presented in
Appendix~\ref{sec:high-order-cons}).

We have
\begin{align*}
  \mathcal{R}_1\colon
  &\phi_1\mapsto-\phi_2,\ \phi_2\mapsto \nu_1,\\
  &\phi_3[\vartheta]\mapsto \phi_4[\vartheta],\ \phi_4[\vartheta]\mapsto
    -\phi_5[\vartheta],\ \phi_5[\vartheta]\mapsto -\phi_6[\vartheta],\
    \phi_6[\vartheta] \mapsto \nu_2[\vartheta],\\
  &\phi_7\mapsto \nu_3,\\
  &\phi_8\mapsto -\phi_9,\ \phi_9\mapsto-\phi_{10},\dots,\phi_{13}\mapsto
    -\phi_{14},\ \phi_{14}\mapsto \nu_4,\\
  &\phi_{15}[\vartheta]\mapsto -2\phi_{16}[\vartheta],\
    \phi_{16}[\vartheta]\mapsto-\phi_{17}[\vartheta],\
    \phi_{17}[\vartheta]\mapsto 
    -\phi_{18}[\vartheta],\ \phi_{18}[\vartheta]\mapsto \nu_5[\vartheta],
    \intertext{and}
    \mathcal{R}_2\colon&\phi_1\mapsto -\nu_1,\ \phi_2\mapsto \nu_6,\\
  &\phi_3[\vartheta]\mapsto-\phi_5[\vartheta],\ \phi_4[\vartheta]\mapsto
    \phi_6[\vartheta],\ \phi_5[\vartheta]\mapsto
    -\nu_2[\vartheta],\ \phi_6[\theta]\mapsto 
    \nu_7[\vartheta],\\
  &\phi_7\mapsto \nu_8,\\
  &\phi_8\mapsto \phi_{10},\dots,\phi_{12}\mapsto \phi_{14},\ \phi_{13} \mapsto
    -\nu_4,\ \phi_{14}\mapsto \nu_9,\\
  &\phi_{15}[\vartheta]\mapsto 2\phi_{17}[\vartheta],\
    \phi_{16}[\vartheta]\mapsto \phi_{18}[\vartheta],\ \phi_{17}[\vartheta]
    \mapsto 
    -\nu_5[\vartheta],\ \phi_{18}\mapsto \nu_{10}[\vartheta].
\end{align*}

\begin{remark}
  We see that the ation of $\mathcal{R}_2$ is equivalent to that
  of~$\mathcal{R}_1\circ\mathcal{R}_1$.
\end{remark}

But this is not the end of a story.

\section{Discussion}
\label{sec:disc-queer-oper}

There exist yet another two nontrivial coverings over~$\mathcal{TE}$ linear
in~$p_\sigma$ and~$q_\sigma$. They are associated with the nonlocal variables
\begin{align*}
  &\begin{array}{rcl}
     w_{5,x} &=& 2u_{xx}p_t - 2u_{xt}p_x - p_{yt},\\
     w_{5,y} &=& 2u_{xy}p_t - u_{yt}p_x - u_yp_{xt} - u_{xt}p_y + u_xp_{yt} -
                 p_{tt} 
   \end{array}
              \intertext{and}
  &\begin{array}{rcl}
     w_{6,x} &=& v_{xx}p_t - v_{xt}p_x + u_{xx}q_t - u_{xt}q_x + q_{yt},\\
     w_{6,y} &=& v_{xy}p_t - 2v_{yt}p_x - 2v_yp_{xt} + v_{xt}p_y + 2v_xp_{yt}
   \\
     &&+ u_{xy}q_t + u_{yt}q_x + u_yq_{xt} - 2u_{xt}q_y - u_xq_{yt} + q_{tt}.
   \end{array}
\end{align*}
The canonical nilpotent field lifts to these coverings by the formulas
\begin{equation}\label{eq:25}
  \dc(w_5) = 2p_xp_t, \qquad \dc(w_6) = p_xq_t - p_tq_x.
\end{equation}
The equation $\tilde{\ell}_{\mathcal{E}}(\Phi) = 0$, where~$\Phi$ may depend
on all the nonlocal variables~$w_1,\dots,w_6$ delivers twelve additional
solutions~$\Phi_3,\dots,\Phi_{14}$ with the components
\begin{align*}
  \Phi_3^u
  &=f (v_{xt} p_t + v_xp_{tt} + u_{xt}q_t + u_xq_{tt} + w_{2,tt} - w_{6,t}),\\
  \Phi_3^v
  &=0,\\[3pt]
  \Phi_4^u
  &=f\left(u_{xt}p_t + u_xp_{tt} - \frac{1}{2}w_{1,tt} -
    \frac{1}{2}w_{5,t}\right),\\ 
  \Phi_4^v
  &=0,\\[3pt]
  \Phi_5^u
  &=-f(yw_{2,tt} + yv_xp_{tt} + yu_xq_{tt} + yu_{xt}q_t + yv_{xt}p_t -
    u_x^2q_t \\
  &- u_x (v_xp_t + w_{2,t} - w_6) - y w_{6,t})
    - y\dot{f}(u_xq_t + v_x p_t + w_{2,t} - w_6),\\
  \Phi_5^v
  &=fv_x(v_x p_t + u_x q_t + w_{2,t} - w_6),\\[3pt]
  \Phi_6^u
  &=f\left(yu_xp_{tt} + yu_{xt}p_t - u_x^2p_t
    + \frac{1}{2}u_x(w_{1,t} + w_5) -
    \frac{1}{2}yw_{5,t}-\frac{1}{2}yw_{1,tt}\right) \\ 
  &+ y\dot{f}\left(u_xp_t - \frac{1}{2}(w_{1,t} + w_5)\right),\\
  \Phi_6^v
  &=\frac{1}{2}f\big(y w_{2,tt} + yv_xp_{tt} + yu_xq_{tt} + y u_xt q_t + y
    v_xt p_t + 2u_x^2 q_t \\
  &- 2u_x(w_6 - w_{2,t}) + v_x(w_{1,t} + w_5) - yw_{6,t}\big) \\
  &+ \frac{1}{2}y\dot{f} (u_xq_t + v_xp_t - w_6 + w_{2,t}),\\[3pt]
  \Phi_7^u
  &=f(v_x p_t + u_x q_t + w_{2,t} - w_6),\\
  \Phi_7^v
  &=0,\\[3pt]
  \Phi_8^u
  &=0,\\
  \Phi_8^v
  &=f(v_{xt}p_t + v_xp_{tt} + u_{xt}q_t + u_xq_{tt} + w_{2,tt} - w_{6,t}),\\[3pt]
  \Phi_9^u
  &=0,\\
  \Phi_9^v
  &=f\left(u_{xt} p_t + u_xp_{tt} - \frac{1}{2}(w_{1,tt} + w_{5,t})\right),\\[3pt]
  \Phi_{10}^u
  &=0,\\
  \Phi_{10}^v
  &=f\left(u_xp_t - \frac{1}{2}(w_{1,t} + w_5)\right),\\[3pt]
  \Phi_{11}^u
  &=f\left(u_xp_t - \frac{1}{2}(w_{1,t} + w_5)\right),\\
  \Phi_{11}^v
  &=0,\\[3pt]
  \Phi_{12}^u
  &=0,\\
  \Phi_{12}^v
  &=\frac{1}{2}f\left(yu_xp_{tt} + yu_{xt}p_t + 2u_x^2p_t - u_x(w_{1,t} + w_5)
    -\frac{1}{2}y (w_{1,tt} + w_{5,t})\right),\\[3pt]
  \Phi_{13}^u
  &=0,\\
  \Phi_{13}^v
  &=f(v_x p_t + u_x q_t + w_{2,t} - w_6),\\[3pt]
  \Phi_{14}^u
  &=0,\\
  \Phi_{14}^v
  &=\frac{1}{2}f\big(yw_{2,tt} + yv_xp_{tt} + yu_xq_{tt} + yu_{xt}q_t +
    yv_{xt}p_t + u_x^2q_t \\
  &+ 2u_x (v_x p_t + w_{2,t} - w_6) - yw_{6,t}\big) 
    + \frac{1}{2}y\dot{f}(u_xq_t + v_x p_t + w_{2,t} - w_6),
\end{align*}
where~$f = f(t)$ and~$\dot{f}$ denotes the $t$-derivative..

Due to Equation~\eqref{eq:25}, these shadows are lifted to the tangent
covering $\mathrm{t}\colon \mathcal{TE} \to\mathcal{E}$. The result is
presented in Appendix~\ref{sec:lifts-phi_3-}. Moreover, the following result
is valid:
\begin{proposition}
  All the shadows $\Phi_3,\dots,\Phi_{11},\Phi_{13},\Phi_{14}$ are lifted to
  the covering over~$\tilde{\tau}\colon \tilde{W}\to\mathcal{TE}$ with the
  nonlocal variables $w_1,\dots,w_6$. The shadow~$\Phi_{12}$ can be lifted if
  and only if $f = \const$. The nonlocal symmetries $\Phi_1$ and $\Phi_2$ are
  also lifted to~$\tilde{\tau}$. One
  has
  \begin{equation}\label{eq:26}
    \ldb \Phi_i,\Phi_j\rdb = 0, \qquad i,j=1,\dots,14,
  \end{equation}
  for these lifts.
\end{proposition}
\begin{proof}
  The explicit expressions for the lifts are given in
  Appendix~\ref{sec:lifts-phi_3-}. The proof of~\eqref{eq:26} is a
  straightforward computation.
\end{proof}

\subsection{``Queer'' operators}
\label{sec:queer-operators}

The recursion operators associated with the shadows~$\Phi_3,\dots,\Phi_{14}$
are extremely degenerate. Namely, their action is as follows: the
operators~$\mathcal{R}_3$, $\mathcal{R}_4$, $\mathcal{R}_7$,
and~$\mathcal{R}_{10}$ take the entire algebra~$\sym(\mathcal{E})$ to the
symmetry~$\phi_3[\vartheta]$. In a similar way,
\begin{align*}
  \mathcal{R}_5\colon&\sym(\mathcal{E})\to \phi_4[\vartheta],\\
  \mathcal{R}_6\colon&\sym(\mathcal{E})\to \phi_4[\vartheta] +
                       \phi_{16}[\vartheta'],\\
  \mathcal{R}_8,\mathcal{R}_9, \mathcal{R}_{11},
  \mathcal{R}_{13}\colon&\sym(\mathcal{E})\to \phi_{15}[\vartheta],\\
  \mathcal{R}_{12}, \mathcal{R}_{14}\colon&\sym(\mathcal{E})\to
                                            \phi_{16}[\vartheta]. 
\end{align*}
Perhaps, this phenomenon is partially explained by the nature of the nonlocal
variable~$w_5$ and~$w_6$, which we, in particular, discuss below.

\subsection{On the theory of two-component conservation laws}
\label{sec:discussion}

Let us study the coverings~$\tau_5$ and~$\tau_6$ in more detail. First of all,
easy computations show that their covering equations, as well as the one for
the Whitney product~$\tau_5\oplus\tau_6$, are differentially connected, i.e.,
all the three coverings are irreducible.

On the other hand, applying formula~\eqref{eq:27} to the conservation laws
corresponding to the coverings at hand, one sees that
\begin{equation}\label{eq:29}
  d_h(\omega_5) = -D_t(F_p)\,dx\wedge\,dy\wedge\,dt,\quad  d_h(\omega_6) =
  D_t(F_q)\,dx\wedge\,dy\wedge\,dt, \qquad \text{ on } J^\infty(\pi),
\end{equation}
where~$F_p$ and~$F_q$ are the 1st and 2nd equations in~\eqref{eq:28},
respectively, which means that~$\psi_{\omega_5} = \psi_{\omega_6} = 0$, i.e.,
our conservation laws are trivial.

Thus, we see that in the multi-dimensional case relations between
two-component conservation laws and the corresponding coverings are more
complicated, than in the case~$\dim M = 2$ (cf.~\cite{VinKrasTrends}). The
following construction is to explain the difference.

Recall that all the conservation laws above were of the form
\begin{equation}\label{eq:30}
  \omega = (X\,dx + Y\,dy)\wedge\,dt
\end{equation}
and consider the subdistribution~$\mathcal{Z}$ in the Cartan distribution
on~$\mathcal{E}$ spanned by the total derivatives~$D_x$ and~$D_y$. Note
that~$\mathcal{Z}$ is obviously Frobenius integrable. Consequently, one can
literally repeat Vinogradov's construction of the $\mathcal{C}$-spectral
sequence, see~\cite{AMV-SecCal}, using the distribution~$\mathcal{Z}$ instead
of~$\mathcal{C}$. Denote this spectral sequence
by~$\{E_s^{p,q}(\mathcal{Z})\}$. The following statement is actually a
reformulation of Vinogradov's results for the case of~$\mathcal{Z}$:

\begin{theorem}
  Consider the system~$\mathcal{E}$ consisting of Equations~\eqref{eq:19}
  and~\eqref{eq:28}. Then\textup{:}
  \begin{enumerate}
  \item Equivalence classes of coverings associated with the
    form~\eqref{eq:30} are in one-to-one correspondence with elements of the
    group~$E_1^{0,1}(\mathcal{Z})$. 
  \item Define the \emph{generating element}~$\psi_\theta$ of a
    form~$\theta\in \ker d_0\subset E_0^{0,1}(\mathcal{Z})$ as the image of
    its coset under the differential
    $d_1\colon E_1^{0,1}(\mathcal{Z})\to E_1^{1,1}(\mathcal{Z})$. Then two
    forms~$\theta$ and~$\theta'$ define equivalent coverings if and only
    if~$\psi_\theta = \psi_{\theta'}$.
  \end{enumerate}
\end{theorem}

\begin{remark}
  To clarify the structure of the spectral sequence at hand, let us discuss
  the following construction. Consider the infinite set of unknowns~$u_0 = u$,
  $u_k = u_{\underbrace{{}_{t\dots t}}_{k \text{ times}}}$, $k\geq 1$, and similar
  for~$v$, $p$, and~$q$. Consider also the infinite system of equations
  \begin{equation}\label{eq:31}
    \begin{array}{rcl}
      u_{k,yy} &=& D_t^k(u_{1,x} - u_x u_{xy} + u_yu_{xx}),\\[2pt]
      v_{k, yy} &=& D_t^k(v_{1,x} - u_x v_{xy} - 2 v_y u_{xx} + 2 v_x u_{xy} + u_y
                 v_{xx}),\\[2pt]
      p_{k, yy} &=& D_t^k(u_yp_{xx} - u_{xy}p_x - u_xp_{xy} + u_{xx}p_y +
                    p_{1,x}),\\ [2pt]
      q_{k,yy} &=& D_t^k(2v_xp_{xy} + v_{xx}p_y - v_{xy}p_x - 2v_yp_{xx} +
                   2u_{xy}q_x \\ 
             &+&   u_yq_{xx} - u_xq_{xy} - 2u_{xx}q_y + q_{1,x})
    \end{array}
  \end{equation}
  evidently obtained from~\eqref{eq:19} and~\eqref{eq:28} and add to it two
  formal relations
  \begin{equation}\label{eq:32}
    t_x = 0, \qquad t_y = 0.
  \end{equation}
  This is a two-dimensional system; denote its infinite prolongation
  by~$\tilde{\mathcal{E}}$. Then the Vinogradov $\mathcal{C}$-spectral
  sequence of~$\tilde{\mathcal{E}}$ coincides
  with~$\{E_s^{p,q}(\mathcal{Z})\}$. In particular, conservation laws
  of~$\tilde{\mathcal{E}}$ coincide with two-component conservation laws of
  the initial equation~$\mathcal{E}$.

  Let us compute the generating elements of~$\omega_5$ and~$\omega_6$. Note
  first that cosymmetries of~$\tilde{\mathcal{E}}$ (elements
  of~$E_1^{1,1}(\mathcal{Z})$) are infinite-dimensional
  covectors. Nevertheless, it is more convenient to present them as infinite
  matrices
  \begin{equation*}
    \psi =
    \begin{vmatrix}
      \psi_0^u&\psi_0^v&\psi_0^p&\psi_0^q\\
      \psi_1^u&\psi_1^v&\psi_1^p&\psi_1^q\\
      \hdotsfor{4}\\
      \psi_k^u&\psi_k^v&\psi_k^p&\psi_k^q\\
      \hdotsfor{4}
    \end{vmatrix},
  \end{equation*}
  where the lines correspond to equations in~\eqref{eq:31}, while the
  superscripts indicate the ``type'' of equation (obviously, the components
  corresponding to Equations~\eqref{eq:32} always vanish). Then we have
  \begin{equation*}
    \psi_{\omega_5} =
    \begin{vmatrix}
      0&0&0&0\\
      0&0&-1&0\\
      0&0&0&0\\
      \hdotsfor{4}
    \end{vmatrix},\qquad
    \psi_{\omega_6} =
    \begin{vmatrix}
      0&0&0&0\\
      0&0&0&1\\
      0&0&0&0\\
      \hdotsfor{4}
    \end{vmatrix},
  \end{equation*}
  so, the corresponding coverings are nontrivial indeed.
\end{remark}

The above observations lead to the following construction. Consider an
$n$-dimensional equation~$\mathcal{E}$ and a horizontal form~$\rho \in
\Lambda_h^{n-2}(\mathcal{E})$. Assume that
\begin{enumerate}
\item $d_h(\rho) = 0$,
\item $\rank\ker\rho = 2$, where~$\rho$ is understood as a map
  $\mathcal{C}(\mathcal{E}) \to \Lambda_h^{n-3}(\mathcal{E})$,
  $X\mapsto i_X(\rho)$.
\end{enumerate}
Then $\mathcal{Z} = \mathcal{Z}_\rho = \ker\rho$ is a Frobenius integrable
$2$-dimensional sub-distribution in~$\mathcal{C}(\mathcal{E})$. Indeed,
let~$X$, $Y\in \mathcal{Z}$. Then
\begin{multline*}
  i_{[X,Y]}(\rho) = (i_X\circ L_Y - L_Y\circ i_X)(\rho) = (i_X\circ
  L_Y)(\rho)\\
  = i_X\circ(d\circ i_Y + i_Y\circ d)(\rho) = (i_X\circ d)(i_Y(\rho)) +
  (i_X\circ i_Y)d(\rho)) = 0,
\end{multline*}
where $i$ and $L$ denote the inner product and Lie derivative, respectively.

Similar to the example above, consider the spectral
sequence~$\{E_s^{p,q}(\mathcal{Z_\rho})\}$ associated
with~$\mathcal{Z}_\rho$. Elements of~$\{E_1^{0,1}(\mathcal{Z_\rho})\}$ are called
two-component conservation laws of type~$\rho$. To any of them a covering
over~$\mathcal{E}$ is associated and this covering is nontrivial if and only
if~$\rho\neq 0$.

\section*{Acknowledgments}
\label{sec:acknowledgements}

Computations were done using the \textsc{Jets} software~\cite{Jets}. The
author is indebted to O.~Morozov for discussions and valuable comments.

\appendix

\section{Explicit formulas for symmetries}
\label{sec:expl-form-symm}

Note fist that System~\eqref{eq:19} covers the first equation by $(u,v)
\mapsto u$. The symmetries $\phi_1,\dots,\phi_6$ are the lifts from
\begin{equation*}
  u_{yy} = u_{xt} - u_xu_{xy} + u_yu_{xx}
\end{equation*}
and are of the form
\begin{align*}
  \phi_1^u
  & = yu_x - 2x,\\
  \phi_1^v
  & = yv_x,\\[3pt]
  \phi_2^u
  & = 2xu_x + yu_y - 3u,\\
  \phi_2^v
  & = 2xv_x + yv_y,\\[3pt]
  \phi_3^u[\vartheta]
  & = \vartheta,\\
  \phi_3^v[\vartheta]
  & = 0,\\[3pt]
  \phi_4^u[\vartheta]
  & = u_x\vartheta - y\vartheta_t,\\
  \phi_4^v[\vartheta]
  & = v_x\vartheta,\\[3pt]
  \phi_5^u[\vartheta]& = u_y\vartheta + (yu_x - x)\vartheta_t -
                       \frac{1}{2}y^2\vartheta_{tt},\\ 
  \phi_5^v[\vartheta]
  & = v_y\vartheta + yv_x\vartheta_t,\\[3pt]
  \phi_6^u[\vartheta]
  & = u_t\vartheta + (xu_x + yu_y - u)\vartheta_t
    + \frac{1}{3}(y^2u_x - 2xy)\vartheta_{tt} - \frac{1}{6}y^3\vartheta_{ttt},\\ 
  \phi_6^v[\vartheta]
  & = v_t\vartheta + (v_x x + v_y y + 2
    v)\vartheta_t + \frac{1}{2}y^2v_x\vartheta_{tt}. 
\end{align*}
The symmetries $\phi_7,\dots,\phi_{18}$ have zero shadows: their
$u$-components vanish, while the $v$-components are
\begin{align*}
  \phi_7^v& = v,\\
  \phi_8^v& = u_{xxx},\\
  \phi_9^v& =  u_xu_{xxx} + u_{xxy} +\frac{1}{2}u_{xx}^2,\\
  \phi_{10}^v& = u_{xxt} + u_x u_{xxy} + u_xu_{xx}^2 + u_{xy}u_{xx} + (u_x^2 +
               u_y) u_{xxx},\\ 
  \phi_{11}^v& = u_x(u_x^2 + 2 u_y)u_{xxx} + (u_x^2 + u_y) u_{xxy} +
               2u_xu_{xxt} + \frac{1}{2}(3 u_x^2 + 2 u_y)u_{xx}^2\\
          &+ (2u_xu_{xy} + u_{xt})u_{xx} + \frac{1}{2}u_{xy}^2 + u_{xyt},\\
  \phi_{12}^v& = (3u_x^2 + 2 u_y)u_{xxt} + u_x (u_x^2 + 2u_y)u_{xxy} + u_{xtt}
               + (2u_x^3 + 3u_xu_y)u_{xx}^2 \\
          &+ u_{xx}(u_{yt} + 3u_xu_{xt} + (3u_x^2 + 2u_y)u_{xy}) + (u_x^4 + 3 u_x^2
            u_y + u_y^2)u_{xxx}\\
          &+
            u_xu_{xy}^2 + u_{xt}u_{xy} + 2u_xu_{xyt},\\ 
  \phi_{13}^v& = (u_x^4 + 3u_x^2 u_y + u_y^2)u_{xxy}
               + (48u_x^3 + 6u_xu_y)u_{xxt} + 3u_xu_{xtt} + u_{ytt}\\
          &+ \frac{1}{2}(5u_x^4 + 12u_x^2 u_y + 3u_y^2)u_{xx}^2\\
            &+((4u_x^3 + 6u_xu_y)u_{xy} + u_{xt} (6u_x^2 + 3u_y) + 3u_{yt}
 u_x)u_{xx} \\
          &+ (u_x^5 + 4u_x^3 u_y + 3u_x u_y^2)u_{xxx} +(3u_x^2 + 2u_y)u_{xyt}
            + \frac{1}{2}(3u_x^2 + 2u_y)u_{xy}^2\\
  &+ u_{xy}(3u_xu_{xt} + u_{yt}) + \frac{3}{2}u_{xt}^2,\\
  \phi_{14}^v& = (5u_x^4 + 12u_x^2u_y + 3u_y^2)u_{xxt}
               + u_x(u_x^2 + u_y)(u_x^2 + 3u_y)u_{xxy} \\
  &+ (6u_x^2 + 3u_y)u_{xtt} + 3u_xu_{ytt} + (3u_x^5 + 10u_x^3u_y +
    6u_xu_y^2)u_{xx}^2\\
          &+ ((5u_x^4 + 12u_x^2u_y + 3u_y^2)u_{xy} + (10u_x^3 + 12u_xu_y)u_{xt}
            + 3(2u_x^2 + u_y)u_{yt})u_{xx}\\
  &+ (u_x^6 + 5u_x^4u_y + 6u_x^2u_y^2 + u_y^3)u_{xxx} + (4u_x^3 +
    6u_xu_y)u_{xyt}
    + u_{ttt}\\
  &+ (2u_x^3 + 3u_xu_y)u_{xy}^2 + u_{xy}(3u_{yt}u_x + u_{xt}(6u_x^2 + 3 u_y))
    + 6 u_xu_{xt}^2 + 3u_{yt}u_{xt},\\ 
  \phi_{15}^v[\vartheta]&= \vartheta,\\
  \phi_{16}^v[\vartheta]&= u_x\vartheta + \frac{1}{2}y\vartheta_{t},\\
  \phi_{17}^v[\vartheta]&= \left(\frac{3}{2}u_x^2 + u_y\right)\vartheta
                          + \left(y u_x + \frac{1}{2}x\right)\vartheta_{t}
                          + \frac{1}{4}y^2\vartheta_{tt},\\
  \phi_{18}^v[\vartheta]&= (2u_x^3 + 3u_xu_y + u_t)\vartheta
                          + \left(\frac{3}{2}yu_x^2 + xu_x + yu_y
                          + \frac{1}{2}u\right)\vartheta_{t}\\
          &+ \frac{1}{2}y(yu_x + x)\vartheta_{tt}
            + \frac{1}{12}y^3\vartheta_{ttt},
\end{align*}
where $\vartheta$ is an arbitrary function in~$t$.

\section{Explicit formulas for lifts}
\label{sec:expl-form-lifts}

The $w_i$-components, $i=1,\dots,4$, of~$\Ev_{\phi_1}$ are
\begin{align*}
  \phi_1^{w_1}
  &= 2u_xpp_x + pp_y,\\
  \phi_1^{w_2}
  &= -v_xpp_x - u_xpq_x - 2u_xp_xq + pq_y + 2p_yq,\\
  \phi_1^{w_3}
  &= -\frac{1}{2}ypp_{yt} + \frac{1}{2}(2yu_{xx} - 1)pp_t
    + \frac{1}{2}(3u_x^2 - 2yu_{xt} - 2u_y)pp_x + yu_xp_tp_x\\
  &+ u_xpp_y + \frac{1}{2}yp_tp_y,\\
  \phi_1^{w_4}
  &= -(yv_{xt} - 2v_y)pp_x + (yu_{xx} - 2)pq_t - (yu_{xt} + u_y)pq_x
    + (2yu_{xx} - 1)p_tq\\
  &- (2yu_{xt} - u_y)p_xq + yu_xp_tq_x - yu_xp_xq_t - 2yv_xp_tp_x + yv_{xx}p
    p_t\\
  &+ u_xpq_y + 2u_xp_yq + v_xpp_y - 3u_{xy}pq + ypq_{yt} + 2yp_{yt}q +
    2yp_tq_y + yp_yq_t. 
\end{align*}
The $w_i$-components of~$\Ev_{\phi_2}$ are
\begin{align*}
  \phi_2^{w_1}
  &= -u_ypp_x - u_xpp_y + 2u_xp_xw_1 - pp_t + p_yw_1,\\
  \phi_2^{w_2}
  &= 2(3u_x^2 - u_y)p_xq + 2u_xpq_y - 2u_xp_yq + 2u_xp_xw_2 -
    u_xq_xw_1  - u_ypq_x\\
  &- v_xpp_y - v_xp_xw_1 + 2v_ypp_x - pq_t - 2p_tq - 2p_yw_2 + q_yw_1,\\
  \phi_2^{w_3}
  &= -\frac{1}{2}(u_x^2 - yu_{xt} - 2u_y)pp_y - \frac{1}{2}y(2 u_x^2 +
    u_y)p_tp_x
    + \frac{1}{2}(2yu_{xx} + 1)p_tw_1\\
  &+ \frac{1}{2}(3u_x^2 - 2yu_{xt})p_xw_1 - (yu_{xy} + u_x)pp_t +
    \frac{1}{2}yu_{yt}pp_x + \frac{1}{2}yu_y pp_{xt} \\
  &- \frac{1}{2}yu_xpp_{yt} - \frac{3}{2}yu_xp_tp_y - u_{xy}pw_1 +
    \frac{1}{2}ypp_{tt}  - \frac{1}{2}yp_{yt}w_1 - \frac{1}{2}yp_yw_{1,t}\\
  &+ u_xp_yw_1 + \frac{1}{2}pw_{1,t},\\ 
  \phi_2^{w_4}
  &= u_xq_yw_1 + yq_yw_{1,t} - 2(yu_{xx} - 1)p_tw_2
    - (6yu_xu_{xx} + 2yu_{xy} - u_x)p_tq \\
  &- (3u_yv_x - 2yv_{yt})pp_x - (4u_x^2 + 2yu_{xt} + u_y)p_yq
    - (yv_{xy} + v_x)pp_t \\
  &+ (2u_x - yu_{xy})pq_t - (u_xv_x + yv_{xt} + 2v_y)pp_y + (yu_{xx} -
    1)q_tw_1\\
  &+ (u_x^2 + 2yu_{xt} + u_y)pq_y + 3(u_xu_{xy} + u_yu_{xx} + u_{xt})pq\\
  &+ (6yu_xu_{xt} - 3u_xu_y - 2yu_{yt})p_xq - 2yp_{yt}w_2 - v_{xy}pw_1
    + 2 u_{xy}pw_2\\
  &- 2yp_yw_{2,t} - u_{xy}qw_1 + yq_{yt}w_1 - ypq_{tt} + v_xp_yw_1 - 2yp_{tt}q
    - 3yp_tq_t \\
  &- 2u_xp_yw_2 + yv_{xx}p_tw_1 - 2yu_yp_{xt}q - yu_{yt}pq_x - y(u_x^2 +
    2u_y)p_tq_x\\
  &- yu_{xt}q_xw_1 + 2yu_{xt}p_xw_2 - 4yu_xp_{yt}q - yv_{xt}p_xw_1
    + y(u_x^2 - u_y)p_xq_t\\
  &+ yu_xpq_{yt} - 2yv_xpp_{yt} - yu_ypq_{xt} + 2yv_ypp_{xt} - qw_{1,t} +
    2pw_{2,t}\\
  &- 3yu_xp_yq_t + 4y(u_xv_x + v_y)p_tp_x.
\end{align*}

\section{``Conventional'' presentation}
\label{sec:conv-pres}

In conventional notation the second operator from
Section~\ref{sec:recurs-oper-their} is
\begin{align*}
  &u_{xx}D_x^2(\tilde{\phi}^u) - u_{xxx}D_x(\tilde{\phi}^u)\\
  &\qquad =(u_x^2 + u_y) u_{xx}D_x^2(\phi^u) + u_xu_{xx}D_xD_y(\phi^u) +
    u_{xx}D_xD_t(\phi^u)\\ 
  &\qquad- (u_x^2 + u_y)u_{xxx}D_x(\phi^u) - u_xu_{xxx}D_y(\phi^u) -
    u_{xxx}D_t(\phi^u)\\ 
  &\qquad+ (u_{xy}u_{xxx} - u_{xx}u_{xxy})\phi^u,\\
  &u_{xx}D_xD_y(\tilde{\phi}^u) - u_{xxy}D_x(\tilde{\phi}^u)\\
  &\qquad= u_xu_yu_{xx}D_x^2(\phi^u) + u_yu_{xx}D_xD_y(\phi^u) +
    u_xu_{xx}D_xD_t(\phi^u) + 
    u_{xx}D_yD_t(\phi^u)\\
  &\qquad+ \big(u_{xt}u_{xx} - (u_x^2 + u_y)u_{xxy}\big)D_x(\phi^u) -
    u_xu_{xxy}D_y(\phi^u)  \\
  &\qquad- (u_{xx}^2 + u_{xxy}) D_t(\phi^u)
    + \big((u_xu_{xx} + u_{xy})u_{xxy} - u_{xx}(u_yu_{xxx} + u_{xxt})\big)\phi^u,\\
  &u_{xx}D_xD_t(\tilde{\phi}^u) - (u_{xx}u_{xy} + u_{xxt})D_x(\tilde{\phi}^u)
    + u_{xx}^2D_y(\tilde{\phi}^u)\\ 
  &\qquad= (u_x^2 + u_y)u_{xx}D_xD_t(\phi^u) + u_xu_{xx}D_yD_t(\phi^u) +
    u_{xx}D_tD_t(\phi^u)\\ 
  &\qquad- \big((u_x^2 + u_y)u_{xxt} + u_{xx}
    ((u_x^2 + u_y)u_{xy} - u_yu_xu_{xx} - 2u_xu_{xt} - u_{yt})\big)D_x(\phi^u) \\
  &\qquad- \big(u_xu_{xxt} + u_{xx}(u_xu_{xy} - u_yu_{xx} -
    u_{xt})\big)D_y(\phi^u)\\
  &\qquad+ \big(u_{xxt} + u_{xx}(u_xu_{xx} + 2u_{xy})\big) D_t(\phi^u)\\
  &\qquad+ \big(u_{xy}u_{xxt} - u_{xx}(u_yu_{xx}^2 + (u_{xt} - u_xu_{xy})u_{xx} +
    u_{xyt} - u_{xy}^2)\big)\phi^u,\\
  &u_{xx}v_{xx}D_x^2(\tilde{\phi}^u) - u_{xx}^2D_x^2(\tilde{\phi}^v)  +
    u_{xx}u_{xxx}D_x(\tilde{\phi}^v) 
    + (u_{xx}v_{xxx} - 2v_{xx}u_{xxx})D_x(\tilde{\phi}^u)\\
  &\qquad =\big((4u_xv_x + 2v_y)u_{xx}^2 + (u_x^2 +
    u_y)v_{xx}u_{xx}\big)D_x^2(\phi^u) \\
   &\qquad + (u_xu_{xx}v_{xx} + 2v_xu_{xx}^2)D_xD_y(\phi^u) +
     v_{xx}u_{xx}D_xD_t(\phi^u)\\ 
  &\qquad+ \big(6v_xu_{xx}^3 + (6u_xv_{xx} + 3v_{xy})u_{xx}^2 +
    ((u_x^2 + u_y)v_{xxx}-2(2u_xv_x + v_y)u_{xxx})u_{xx}\\
  &\qquad- 2(u_x^2 + u_y)v_{xx}u_{xxx}\big)D_x(\phi^u) \\
    &\qquad+ \big(3v_{xx}u_{xx}^2 + (u_xv_{xxx} - 2v_xu_{xxx})u_{xx} -
    2u_xv_{xx}u_{xxx}\big) D_y(\phi^u)\\
  &\qquad+ (u_{xx}v_{xxx} - 2v_{xx}u_{xxx})D_t(\phi^u)\\
  &\qquad+ \big(u_{xx}^2v_{xxy} - v_{xx}u_{xx}u_{xxy}
    - (v_{xy}u_{xxx} + u_{xy}v_{xxx})u_{xx} + 2v_{xx}u_{xy}u_{xxx}\big) \phi^u\\
  &\qquad- (u_x^2 + u_y)u_{xx}^2D_x^2(\phi^v) - u_xu_{xx}^2D_xD_y(\phi^v) -
    u_{xx}^2D_xD_t(\phi^v)\\ 
  &\qquad+ \big((u_x^2 + u_y)u_{xx}u_{xxx}-6u_xu_{xx}^3 -
    3u_{xx}^2u_{xy}\big)D_x(\phi^v)  \\
  &\qquad+ (u_xu_{xx}u_{xxx} - 3u_{xx}^3)D_y(\phi^v) +
    u_{xx}u_{xxx}D_t(\phi^v)\\
  &\qquad+ (2u_{xx}u_{xy}u_{xxx}-6u_{xx}^4 - 2u{_xx}^2 u_{xxy}) \phi^v ,\\
  &u_{xx}v_{xx}D_xD_y(\tilde{\phi}^u) - u_{xx}^2D_xD_y(\tilde{\phi}^v)
    + (u_{xx}v_{xxy} - 2v_{xx}u_{xxy})D_x(\tilde{\phi}^u) +
    u_{xx}u_{xxy}D_x(\tilde{\phi}^v)\\ 
  &\qquad= u_{xx}\big(2(u_xv_y + u_yv_x)u_{xx} + u_xu_yv_{xx}\big)D_x^2(\phi^u)
    + u_{xx}(2v_yu_{xx} + u_yv_{xx}) D_xD_y(\phi^u)\\
  &\qquad+ u_x{x} (u_xv_{xx} + 2v_xu_{xx})D_xD_t(\phi^u) +
    u_{xx}v_{xx}D_yD_t(\phi^u)\\ 
  &\qquad+ \big(u_{xx}((u_x^2 + u_y)v_{xxy}
    - (3u_xv_{xy} + 6v_xu_{xy} + 3u_yv_{xx} + 2v_{xt})u_{xx} - v_{xx}u_{xt})\\
  &\qquad-2((2u_xv_x + v_y)u_{xx} + (u_x^2 + u_y)v_{xx})u_{xxy}\big) D_x(\phi^u) \\
  &\qquad- \big(2(u_xv_{xx} + v_xu_{xx})u_{xxy} - u_{xx}(u_xv_{xxy} +
    3u_{xx}v_{xy})\big)D_y(\phi^u)\\
  &\qquad+ (u_{xx}v_{xxy} - 2v_{xx}u_{xxy})D_t(\phi^u)
    + \big((2v_xu_{xx}^2 + (u_xv_{xx} - v_{xy})u_{xx} +
    2u_{xy}v_{xx})u_{xxy}\\
  &\qquad- u_{xx}((u_xu_{xx} + u_{xy})v_{xxy} +
    v_{xx}u_{xxt} - u_{xx}v_{xxt} + 3v_{xy}u_{xx}^2\\
  &\qquad+ (2v_yu_{xxx} - 3u_{xy}v_{xx} - u_yv_{xxx})u_{xx} +
    u_yv_{xx}u_{xxx})\big) \phi^u \\ 
  &\qquad- u_yu_xu_{xx}^2D_x^2(\phi^v) - u_xu_{xx}^2D_xD_t(\phi^v) -
    u_yu_{xx}^2D_xD_y(\phi^v) -  
    u_{xx}^2D_yD_t(\phi^v)\\ 
  &\qquad+ \big((u_x^2 + u_y)u_{xx}u_{xxy}
    - u_{xx}(3u_yu_{xx}^2 + (3u_xu_{xy} + u_{xt})u_{xx})\big)D_x(\phi^v)\\
  &\qquad+ (u_xu_{xx}u_{xxy} - 3u_{xx}^2u_{xy})D_y(\phi^v)
    + (u_{xx}u_{xxy} - 2u_{xx}^3)D_t(\phi^v)\\
  &\qquad+ \big((2u_xu_{xx}^2 + 2u_{xx}u_{xy})u_{xxy}
    - u_{xx}(6u_{xx}^2u_{xy} + 2u_yu_{xx}u_{xxx} + 2u_{xx}u_{xxt})\big)\phi^v,\\
  &u_{xx}v_{xx}D_xD_t(\tilde{\phi}^u) -
    (2v_{xx}u_{xxt} + (2u_{xx}v_{xy} - 2u_{xy}v_{xx} -
    v_{xxt})u_{xx})D_x(\tilde{\phi}^u)\\ 
  &\qquad- u_{xx}^2D_xD_t(\tilde{\phi}^u) + (u_{xx}u_{xxt} -
    2u_{xx}^2u_{xy})D_x(\tilde{\phi}^v) + 
    2u_{xx}^3D_y(\tilde{\phi}^v)\\
  &\qquad= u_{xx}\big((4u_xv_x + 2v_y)u_{xx} + (u_x^2 +
    u_y)v_{xx}\big)D_xD_t(\phi^u)\\ 
    &\qquad+ u_{xx}(u_xv_{xx} + 2v_xu_{xx})D_yD_t(\phi^u) +
      u_{xx}v_{xx}D_tD_t(\phi^u)\\ 
  &\qquad+ \big((u_x^2 + u_y) u_{xx}v_{xxt}
    -2((2u_x v_x + v_y) u_xx + (u_x^2 + u_y)v_{xx}) u_{xxt}\\
  &\qquad -4(u_xv_y + u_yv_x)u_{xx}^3\\
  &\qquad+ 2(4u_xv_xu_{xy}-u_x^2v_{xy} + 2u_xv_{xt} + 2v_xu_{xt} + 2v_yu_{xy}  - u_y
    v_{xy} + v_{yt})u_{xx}^2\\
  &\qquad+ (2u_x^2u_{xy}v_{xx} + 2u_xu_{xt}v_{xx} + 2u_yu_{xy}v_{xx} + u_{yt}v_{xx})
    u_{xx} \big)D_x(\phi^u)\\
  &\qquad+ \big(u_xu_{xx}v_{xxt} - 2(u_xv_{xx} + v_xu_{xx})u_{xxt}
    - 4v_yu_{xx}^3 +
    2(2v_xu_{xy} -u_xv_{xy} + v_{xt})u_{xx}^2\\
  &\qquad+ (2u_xu_{xy}v_{xx} + u_{xt}v_{xx}) u_xx\big) D_y(\phi^u)\\
  &\qquad+(u_{xx}v_{xxt} - 2v_{xx}u_{xxt}-2v_xu_{xx}^3 - u_{xx}^2 v_{xy} +
    u_{xx}u_{xy}v_{xx}) D_t(\phi^u) \\
  &\qquad+
    \big((2u_{xy}v_{xx}-u_{xx}v_{xy})u_{xxt} - u_{xx} u_{xy} v_{xxt}
    + 4 v_yu_{xx}^4\\
  &\qquad+
    2(u_xv_{xy} - 2v_xu_{xy} - u_yv_{xx} - v_{xt})u_{xx}^3
    + (4u_{xy} v_{xy} + v_{xyt})u_{xx}^2 \\
  &\qquad-(2 u_{xy}^2 v_{xx}+ v_{xx}u_{xyt})u_{xx}\big)\phi^u
    -u_{xx}^2 (u_x^2 + u_y)D_xD_t(\phi^v) -u_xu_{xx}^2 D_yD_t(\phi^v) \\
  &\qquad-u_{xx}^2 D_t^2(\phi^v) + \big((u_x^2 + u_y)u_{xx}u_{xxt} +
    2u_xu_yu_{xx}^3\\ 
    &\qquad- (2 u_x^2u_{xy} + 2 u_xu_{xt} + 2u_yu_{xy} + u_{yt})u_{xx}^2\big)
      D_x(\phi^v)\\ 
  &\qquad+ \big(u_xu_{xx} u_{xxt} + 2u_yu_{xx}^3
    - (2u_xu_{xy} + u_{xt})u_{xx}^2\big)D_y(\phi^v)\\
    &\qquad+ ( u_{xx} u_{xxt}-2 u_x u_{xx}^3 - 4 u_{xx}^2 u_{xy}) D_t(\phi^v)\\
  &\qquad+ 2\big(u_{xy}u_{xx}u_{xxt} + 2u_yu_{xx}^4 - (2u_xu_{xy} +
    u_{xt})u_{xx}^3 - (2u_{xy}^2 + u_{xyt}) u_{xx}^2\big) \phi^v. 
\end{align*}

\section{Nonlocal shadows generated by $\mathcal{R}_1$
  and $\mathcal{R}_2$} 
\label{sec:high-order-cons}

Consider the covering over $\mathcal{E}$ with the nonlocal variables
$z_1,\dots, z_4$ defined by
\begin{align*}
  z_{1,x} &= u_x^2 + u_y,\\
  z_{1,y} &=  u_x u_y + u_t;  \\[3pt]
  z_{2,x} &= u_x v_x - v_y,\\
  z_{2,y} &= 2 u_x v_y - u_y v_x - v_t;    \\[3pt]
  z_{3,x} &= (2 u_x u_{xt} + u_{yt}) y - u_x^3 - 2 u_x u_y - u_t,\\
  z_{3,y} &= (u_x u_{yt} +
            u_{xt} u_y + u_{tt}) y - u_x^2 u_y - u_y^2;    \\[3pt]
  z_{4,x} &= (u_x v_{xt} + u_{xt} v_x - v_{yt})y - u_xv_y - u_yv_x + v_t,\\
  z_{4,y} &=
            u_x^2 v_y - u_y v_y - u_x u_y v_x + (2 u_xv_{yt} + 2 u_{xt}
            v_y - u_yv_{xt} - u_{yt} v_x - v_{tt}) y.
  \end{align*}
Then the nonlocal shadows mentioned in Section~\ref{sec:recurs-oper-their} are
\begin{align*}
  \nu_1^u
  &=-3 u u_x - yu_t - 2 xu_y + 4 z_1,\\
  \nu_1^v
  &=3 u v_x - yv_t - 2 xv_y - z_2;   \\[3pt]
  \nu_2^u[\vartheta]
  &= \frac{1}{24}y^4\vartheta_{tttt}
    - \frac{1}{6}y^2(yu_x - 3x)\vartheta_{ttt}
    - \frac{1}{2}(2u_x x y + u_y y^2 - 2u y - x^2)\vartheta_{tt}\\
  &- (uu_x + yu_t + xu_y - z_1)\vartheta_{t}
    + (u_x u_t - z_{1,t}) \vartheta,\\
  \nu_2^v[\vartheta]
  &=-\frac{1}{6}y^3v_x\vartheta_{ttt}
    - \frac{1}{2}y(2xv_x + yv_y + 4v)\vartheta_{tt}\\
  &- (u v_x + 4 vu_x + yv_t + xv_y - 2z_2)\vartheta_{t}
    + (u_t v_x - 2 u_x v_t + z_{2, t}) \vartheta;   \\[3pt]
  \nu_3^u
  &=0,\\
  \nu_3^v
  &=2vu_x + z_2;\\[3pt]
  \nu_4^u
  &=0,\\
  \nu_4^v
  &= -(u_x^7 + 6 u_x^5 u_y + 10 u_x^3 u_y^2 + 4 u_x u_y^3)u_{xxx}\\
  &- (u_x^6 + 5 u_x^4 u_y + 6 u_x^2 u_y^2 + u_y^3)u_{xxy}
        - (6 u_x^5 + 20 u_x^3 u_y + 12 u_x u_y^2)u_{xxt}\\
  &- (5u_x^4 + 12u_x^2u_y + 3u_y^2)u_{xyt} - (10u_x^3 + 12u_xu_y)u_{xtt}\\
  &- (6u_x^2 + 3u_y)u_{ytt} - (4 u_x - 2 u_{xx})u_{ttt}\\
  &- (\frac{7}{2} u_x^6 + 15 u_x^4 u_y + 15 u_x^2 u_y^2 + 2 u_y^3)u_{xx}^2 \\
  &- (6 u_x^5 + 20 u_x^3 u_y + 12 u_x u_y^2)u_{xx}u_{xy} \\
  &-(15 u_x^4 u_{xt} + 10 u_x^3 u_{yt} + 30 u_x^2u_yu_{xt} -
        12u_xu_yu_{yt} - 6 u_y^2u_{xt})u_{xx} \\ 
  &- (\frac{5}{2}u_x^4 + 6u_x^2u_y + \frac{3}{2}u_y^2)u_{xy}^2 \\
  &- (10u_x^3u_{xt} + 6u_x^2u_{yt} - 12u_xu_yu_{xt}  - 3u_yu_{yt})u_{xy}
  \\
  & - 15u_x^2u_{xt}^2 - 12u_xu_{yt}u_{xt} - 6u_yu_{xt}^2 -
        \frac{3}{2}u_{yt}^2- z_{1,ttt};\\[3pt]
  \nu_5^u[\vartheta]
  &=0,\\
  \nu_5^v[\vartheta]
  &= -\frac{1}{48}y^4\vartheta_{tttt} - (\frac{1}{6}y^3u_x  +
    \frac{1}{4}xy^2)\vartheta_{ttt}\\
  &- (\frac{3}{4} y^2u_x^2 - xyu_x - \frac{1}{2}y^2u_y - \frac{1}{2}yu -
    \frac{1}{4}x^2)\vartheta_{tt} \\ 
  &- (2 y u_x^3 + \frac{3}{2}xu_x^2 + (3u_y y + u)u_x + yu_t
    + xu_y + \frac{1}{2}z_1)\vartheta_{t} \\
  &- (\frac{5}{2}u_x^4 + 6u_x^2u_y + 2u_tu_x + \frac{3}{2}u_y^2 + z_{1,t})
    \vartheta;\\[3pt]
  \nu_6^u
  &= 3 uu_y + 2 x u_t + (2 u_x - 8)z_1 - y(u_xu_t + 2 z_{1, t}),\\
  \nu_6^v
  &=3uv_y + 2xv_t + 2v_xz_1 + (4u_x - 1)z_2 + y(2v_tu_x - u_tv_x + 2z_{2, t});\\
  \nu_7[\vartheta]
  &=\frac{1}{120}y^5\vartheta_{ttttt}
    + \frac{1}{24}y^3(yu_x - 4x)\vartheta_{tttt} \\
  &+ \frac{1}{2}(xy^2u_x + \frac{1}{3}y^3u_y - y^2u - yx^2)\vartheta_{ttt}
  \\
  &+ (\frac{1}{2}(2yu + x^2)u_x + \frac{1}{2}y^2u_t + x(yu_y -
    u))\vartheta_{tt} \\
  &+ (uu_y + u_t - xyu_tu_x - z_1)\vartheta_{t} -
    (u_y u_t - z_{1,t})\vartheta \\
  & + \frac{5}{2}yz_{2,t}- \frac{5}{2}u_xz_3,\\
  \nu_7[\vartheta]
  &=\frac{1}{24}y^4v_x\vartheta_{tttt}
    + y^2(\frac{1}{2}v_x x + \frac{1}{6}yv_y + v)\vartheta_{ttt}\\
  &+ ((yu + \frac{1}{2} x^2)v_x + 4yvu_x + \frac{1}{2}y^2v_t +
    xyv_y)\vartheta_{tt} \\
  &+ (6vu_x^2 - yu_tv_x + 2yu_xv_t +
    u v_y + 4v u_y + xv_t - 2z_2)\vartheta_{t} \\
  &+ (yv_xu_{tt} + 3u_x^2v_t + yu_xv_{tt} - u_tv_y + 2u_yv_t -
    z_{2,t})\vartheta \\
  &- \frac{5}{2}v_xz_3 + 4u_xz_4 + 2yz_{4,t};\\[3pt]
  \nu_8^u
  &=0,\\
  \nu_8^v
  &=3vu_x^2 + 2vu_y - 2u_xz_2 - yz_{2,t} + z_4;\\[3pt]
  \nu_9^u
  &=0,\\
  \nu_9^v
  &=-4yu_xu_{tttt}
    + (u_x^8 + 7u_x^6u_y + 15u_x^4u_y^2 + 10u_x^2u_y^3 + u_y^4)u_{xxx}\\
  &+ (u_x^7 + 6 u_x^5 u_y + 10 u_x^3 u_y^2 + 4 u_x u_y^3)u_{xxy}\\
  &+ (7 u_x^6 + 30 u_x^4 u_y + 30 u_x^2 u_y^2 + 4 u_y^3)u_{xxt}\\
  &+ (6 u_x^5 + 20 u_x^3 u_y + 12 u_x u_y^2)u_{xyt}
    + (15 u_x^4 + 30 u_x^2 u_y + 6 u_y^2)u_{xtt}\\
  &+ (10 u_x^3 + 12 u_x u_y)u_{ytt} + (3 u_x^2 + 2 u_y)u_{ttt}\\
  &+ (4 u_x^7 + 21 u_x^5 u_y + 30 u_x^3 u_y^2 + 10 u_x u_y^3)u_{xx}^2\\
  &+ (7 u_x^6 + 30 u_x^4 u_y + 30 u_x^2 u_y^2 + 4 u_y^3)u_{xx}u_{xy}\\
  &+ ((21 u_x^5 + 60 u_x^3 u_y + 30 u_x u_y^2))u_{xx}u_{xt}\\
  &+ (15 u_x^4 + 30 u_x^2 u_y + 6 u_y^2)u_{xx}u_{yt} + (3 u_x^5 + 10 u_x^3
    u_y + 6 u_x u_y^2)u_{xy}^2 \\
  &+ (15 u_x^4 + 30 u_x^2 u_y + 6 u_y^2)u_{xy}u_{xt}
    + (10 u_x^3 + 12 u_x u_y)u_{xy}u_{yt} \\
  &+ (30 u_x^3 + 30 u_x u_y)u_{xt}^2 + (30 u_x^2 + 12 u_y)u_{yt}u_{xt}
    + 6 u_x u_{yt}^2\\
  &+ yz_{1,tttt} + 2u_xz_{1, ttt} - z_{3, ttt};\\[3pt]
  \nu_{10}^u
  &=0,\\
  \nu_{10}^v
  &=\frac{1}{240}y^5\vartheta_{ttttt}
    + \frac{1}{24}y^3(yu_x + 2x)\vartheta_{tttt}\\
  &+ \frac{1}{4}y((u_x^2 + \frac{2}{3} u_y) y^2
    + (2xu_x + u) y + x^2)\vartheta_{ttt} \\
  &+ \frac{1}{2}\big((2 u_x^3 + 3 u_x u_y + u_t) y^2 +
    (3xu_x^2 + 2 u u_x + 2xu_y + z_1) y\\
  &+ x (xu_x + u)\big)\vartheta_{tt}
    + \big(\frac{3}{2} yz_{1,t} + \frac{5}{2}yu_x^4 +
    2 x u_x^3 + (6yu_y + \frac{3}{2} u)u_x^2 \\
  &+ (2yu_t + 3xu_y + z_1) u_x + \frac{3}{2}yu_y^2 + uu_y + x u_t -
    \frac{1}{2}z_3\big)\vartheta_t \\
  &+ \big(3u_x^5 + 10u_x^3u_y + 3u_tu_x^2 + 6u_xu_y^2 + 2u_xz_{1,t} +
    yz_{1,tt} \\
  &+ 2 u_t u_y - z_{3,t}\big)\vartheta
\end{align*}

\section{Lifts of $\phi_3,\dots,\phi_{14}$}
\label{sec:lifts-phi_3-}

The components~$\phi_i^p$, $\phi_i^q$, $i = 3,\dots,14$, are
\begin{align*}
  &\phi_3^p=0,
  &&\phi_3^q=0,\\
  &\phi_4^p=0,
  &&\phi_4^q=0,\\
  &\phi_5^p=f(v_x p_t p_x-u_x p_x q_t - p_x w_{2,t}
  &&\phi_5^q= f(u_x q_tq_x + v_xp_t q_x + q_x w_6 - q_xw_{2,t}),\\
  & \qquad+ p_xw_6),
  &&\\
  &\phi_6^p=-\frac{1}{2}f(2 u_xp_t p_x + p_x w_5 + p_x w_{1,t}),
  &&\phi_6^q=f
     \Big(v_x p_t p_x - u_x p_t q_x - u_x p_x q_t \\
  &&&\qquad+ p_xw_6 - \frac{1}{2}q_x w_5 - p_xw_{2,t} - \frac{1}{2}q_x
      w_{1,t}\Big),\\ 
  &\phi_7^p=0,
  &&\phi_7^q=0,\\
  &\phi_8^p=0,
  &&\phi_8^q=0,\\
  &\phi_9^p=0,
  &&\phi_9^q=0,\\
  &\phi_{10}^p=0,
  &&\phi_{10}^q=0,\\
  &\phi_{11}^p=0,
  &&\phi_{11}^q=0,\\
  &\phi_{12}^p=0,
  &&\phi_{12}^q=\frac{1}{2}f(2 u_xp_t p_x + p_x w_5 + p_x w_{1,t}),\\
  &\phi_{13}^p=0,
  &&\phi_{13}^q=0,\\
  &\phi_{14}^p=0,
  &&\phi_{14}^q=f(v_x p_t p_x - u_x p_x q_t - p_x w_{2,t} + p_x w_6).
\end{align*}
The lifts of $\phi_1$ and $\phi_2$ to $\tilde{\tau}$ are
\begin{align*}
  \phi_1^{w_5}
     &=-pp_{yt} + p_tp_y + 2u_{xx}pp_t - 2u_{xt}pp_x + 2u_xp_tp_x,\\
     \phi_1^{w_6}
     &=u_{xx}pq_t - u_xp_xq_t - 2v_xp_tp_x - v_{xt}pp_x + v_{xx}pp_t
         - u_{xt}pq_x + 2u_{xx}p_tq \\
      &+ pq_{yt} - 2u_{xt}p_xq + 2p_tq_y + 2p_{yt}q + p_yq_t + u_xp_tq_x;\\
        \phi_2^{w_5}
     &=-(2u_x^2 + u_y)p_tp_x + u_{xt}pp_y - 2u_{xt}p_xw_1 + 2u_{xx}p_tw_1
         - 2u_{xy}pp_t\\
     &+ u_{yt}pp_x + u_ypp_{xt} - u_xpp_{yt} - 3u_xp_tp_y + pp_{tt}
        - p_{yt}w_1 - p_yw_{1,t},\\
     \phi_2^{w_6}
     &=-4u_xp_{yt}q - v_{xt}pp_y + 2v_ypp_{xt} - u_ypq_{xt} - 2v_xpp_{yt}
         + 2u_{xt}pq_y - u_{yt}pq_x\\
     &+ 2v_{yt}pp_x + u_xpq_{yt} - u_{xy}pq_t + 2(3u_{xt}u_x - u_{yt})p_xq
        + u_{xx}q_tw_1 + q_yw_{1,t}\\
     &- 2p_{tt}q - 3p_tq_t - 2u_{xx}p_tw_2 - 2(3u_xu_{xx} + u_{xy})p_tq
        - 2u_{xt}p_yq - 2p_yw_{2,t}\\
     &- 2p_{yt}w_2 + q_{yt}w_1 - pq_{tt} + 4(u_xv_x + v_y)p_tp_x
        - (u_x^2 + 2 u_y)p_tq_x \\
     &- v_{xt}p_xw_1 + (u_x^2 - u_y) p_x q_t + 2u_{xt}p_xw_2 - u_{xt}q_xw_1
        + v_{xx}p_tw_1\\
     &- 2u_yp_{xt}q - 3u_xp_yq_t - v_{xy}pp_t.
\end{align*}
Finally, the lifts of~$\phi_3,\dots,\phi_{14}$ to $\tilde{\tau}$ are
\begin{align*}
  \phi_3^{w_i}
  &=0,\quad i=1,\dots,6;\\
  \phi_4^{w_i}
  &=0,\quad i=1,\dots,6;\\
  \phi_5^{w_1}
  &=f(2u_xv_xp_tp_x - 2u_x^2p_xq_t - 2u_xp_xw_{2,t} + 2u_xp_xw_6 +
    v_xp_tp_y \\
  &- u_x p_y q_t + p_y w_6 - p_y w_{2,t}),\\
  \phi_5^{w_2}&=-f(u_xv_xp_tq_x + u_x^2q_tq_x + v_x^2p_tp_x - v_xu_x
                p_xq_t - u_xq_xw_{2,t} - v_xp_xw_{2,t}\\
  &+ (u_xq_x + v_xp_x - q_y)w_6 - u_xq_tq_y - v_xp_tq_y + q_yw_{2,t},\\
  \phi_5^{w_3}
  &=\frac{1}{2}\dot{f}(y(u_xp_yq_t - v_xp_tp_y - p_yw_6 + p_yw_{2,t}) - v_xpp_t
    - u_x p q_t + pw_6 - pw_{2,t})\\
  &- f(\frac{1}{2}y(v_xp_{tt}p_y + p_{yt}w_6) - u_xp_yw_6
    + (u_xp_y - u_xy p)w_{2,t}\\
  &+ (yu_{xt} - \frac{3}{2} u_x^2)p_xw_6 - (\frac{1}{2} + yu_{xx})p_tw_6
    + (\frac{3}{2}u_x^2 - yu_{xt})p_xw_{2,t}\\
  &+ (yu_{xx} + \frac{1}{2})p_tw_{2,t}
    + (\frac{3}{2} u_x^3 - yu_{xt}u_x)p_xq_t
    + (u_x^2 - \frac{1}{2}yu_{xt})p_yq_t \\
  &+ (\frac{1}{2}yv_{xt} - u_xv_x)p_tp_y
    + (\frac{1}{2}u_{xt} - u_xu_{xy})pq_t
    + (\frac{1}{2}v_{xt} - u_{xy}v_x)pp_t + \frac{1}{2}u_xpq_{tt}\\
  &+ u_{xy}pw_6 + \frac{1}{2}v_xpp_{tt} + u_x(yu_{xx} + \frac{1}{2})p_tq_t
    + v_x(yu_{xt} - \frac{3}{2}u_x^2)p_tp_x - \frac{1}{2}yu_xp_yq_{tt}\\
  &+ \frac{1}{2}(yv_xp_tp_{yt} - yu_xp_{yt}q_t + pw_{2,tt} - pw_{6,t}
    - yp_{yt}w_{2,t} - yp_yw_{2,tt} + yp_y w_{6,t})),\\
  \phi_5^{w_4}
  &=\dot{f} (yu_xq_tq_y + yv_xp_tq_y + u_xq q_t - v_xp_tq + yq_yw_6
    + (q - yq_y)w_{2,t} - qw_6)\\
  &+ f(t)((u_xu_{xy} + u_{xt})qq_t -(yu_xv_{xx} - yu_{xx}v_xy + v_x)p_tq_t
    - (u_{xy}v_x + v_{xt})p_tq\\
  &- v_xp_{tt}q + u_xqq_{tt} - v_xp_yw_{2,t} - v_{xy}pw_6 - u_xq_yw_{2,t}
    - (yu_xx - 1)q_tw_{2,t} \\
  &+ u_xq_yw_6 + (u_{xy}q - yq_{yt})w_{2,t} + yq_{yt}w_6
    + (u_xv_x + yv_{xt})p_tq_y \\
  &+ v_x^2p_tp_y + (u_x^2 + yu_{xt})q_tq_y + v_{xy}pw_{2,t} + yq_yw_{6,t}
    + v_xp_yw_6 + (yu_{xx} - 1)q_tw_6 \\
  &- yq_yw_{2,tt} + yu_{xt}q_xw_{2,t} + qw_{2,tt} + yv_{xx}p_tw_6
    - v_xu_xp_yq_t + v_{xy}u_xpq_t + v_{xy}v_xpp_t\\
  &- qw_{6,t} - u_{xy}qw_6 - yv_xv_{xt}p_tp_x - yu_xu_{xt}q_tq_x
    - yu_{xt}v_xp_tq_x + yu_xv_{xt}p_xq_t \\
  &- yu_{xt}q_xw_6 + yv_{xt}p_xw_{2,t} - yv_{xt}p_xw_6 - yv_{xx}p_tw_{2,t}
    + yv_xp_tq_{yt} \\
  &+ yu_xq_tq_{yt} + yu_xq_{tt}q_y + yv_xp_{tt}q_y),\\
  \phi_5^{w_5}
  &=\dot{f}(-v_x p_t p_y + u_x p_y q_t - p_y w6 + p_y w2_t)\\
  &- 2f(u_{xt}v_xp_tp_x - u_xu_{xt}p_xq_t + u_{xx}u_xp_tq_t + u_{xt}p_xw_6
    - \frac{1}{2}u_{xt}p_yq_t - u_{xt}p_xw_{2,t} \\
  &- u_{xx}p_tw_6 + u_{xx}p_tw_{2,t} + \frac{1}{2}v_{xt}p_tp_y
    + \frac{1}{2}v_xp_tp_{yt} + \frac{1}{2}v_xp_{tt}p_y
    - \frac{1}{2}u_xp_yq_{tt}\\
  &- \frac{1}{2}u_xp_{yt}q_t + \frac{1}{2}p_{yt}w_6 +
    \frac{1}{2}p_yw_{6,t}
    - \frac{1}{2}p_yw_{2,tt} - \frac{1}{2}p_{yt}w_{2,t}),\\
  \phi_5^{w_6}
  &=\dot{f}(u_xq_tq_y + v_xp_tq_y + q_yw_6 - q_yw_{2,t})\\
  &- f((u_xv_{xx} - u_{xx}v_x)p_tq_t + u_{xt}v_xp_tq_x + u_xu_{xt}q_tq_x
    + v_xv_{xt}p_tp_x - u_xv_{xt}p_xq_t + u_{xt}q_xw_6\\
  &- u_{xt}q_tq_y - u_{xt}q_xw_{2,t} - u_{xx}q_tw_6 + u_{xx}q_tw_{2,t} +
    v_{xt}p_xw_6 - v_{xt}(p_tq_y - p_xw_{2,t})\\
  &- v_{xx}p_tw_6 + v_{xx}p_tw_{2,t} - v_xp_tq_{yt} - v_xp_{tt}q_y -
    u_xq_tq_{yt}\\
  &- u_xq_{tt}q_y - q_{yt}w_6 - q_yw_{6,t} + q_yw_{2,tt} +
    q_{yt}w_{2,t});\\
  \phi_6^{w_1}
  &=-f(2u_x^2p_tp_x + u_xp_tp_y + u_xp_xw_5 + u_xp_xw_{1,t}
    + \frac{1}{2}p_yw_5 + \frac{1}{2}p_yw_{1,t}),\\
  \phi_6^{w_2}
  &=f(u_x^2p_tq_x + u_x^2p_xq_t - u_xp_xw_6 + \frac{1}{2}u_x
    q_xw_5 - u_xp_tq_y - u_xp_yq_t + u_xp_xw_{2,t}\\
  &+ \frac{1}{2}u_xq_xw_{1,t} + \frac{1}{2}v_xp_xw_5 + v_xp_tp_y
    + \frac{1}{2}v_xp_xw_{1,t} + p_yw_6 - \frac{1}{2}q_yw_5 -
    p_yw_{2,t} - \frac{1}{2}q_yw_{1,t}),\\
  \phi_6^{w_3}
  &=\frac{1}{4}\dot{f}(2yu_xp_tp_y + 2u_xpp_t + yp_yw_5 + yp_yw_{1,t} -
    pw5 - pw_{1,t})\\
  &+ \frac{1}{2}f(t)(-(2u_x^2 - yu_{xt})p_tp_y - (2u_xu_{xy} - u_{xt})pp_t
    - \frac{1}{2}(2yu_{xx} + 1)p_tw_{1,t} \\
  &- \frac{1}{2}(3u_x^2 - 2yu_{xt})p_xw_{1,t} - u_xp_yw_{1,t} -
    u_x (3u_x^2 - 2yu_{xt})p_tp_x + u_{xy}pw_{1,t}\\
  &- \frac{1}{2}(2yu_{xx} + 1)p_tw_5 - \frac{1}{2}(3u_x^2 - 2yu_{xt})p_xw_5
    - \frac{1}{2}pw_{1,tt} - \frac{1}{2}pw_{5,t} + yu_xp_tp_{yt}\\
  &+ yu_xp_{tt}p_y + u_{xy}pw_5 + u_xpp_{tt} - u_xp_yw_5 +
    \frac{1}{2}yp_yw_{5,t}
    + \frac{1}{2}yp_yw_{1,tt} + \frac{1}{2}yp_{yt}w_{1,t}
    + \frac{1}{2}y p_{yt}w_5),\\
  \phi_6^{w_4}
  &=\frac{1}{2}\dot{f}(-2yu_xp_tq_y - 2yu_xp_yq_t + 2yv_xp_tp_y + 2u_xpq_t
    + 2u_xp_tq + 2v_xpp_t + 2 yp_yw_6 \\
  &- yq_yw_5 - 2yp_yw_{2,t} - yq_yw_{1,t} - 2pw_6 + qw_5 + 2pw_{2,t} + qw_{1,t})\\
  &- f(t)(yp_{yt}w_{2,t} + yp_yw_{2,tt} - yp_yw_{6,t} - yv_xp_{tt}p_y -
    yp_{yt}w_6
    - u_xp_yw_6 + u_xp_yw_{2,t}\\
  &- u_{xy}pw_{2,t} - u_xpq_{tt} + u_{xy}pw_6 - v_xpp_{tt} - yu_{xt}u_xp_tq_x
    - yu_xu_{xt}p_xq_t + yu_xp_yq_{tt}\\
  &+ yu_xp_{yt}q_t - yv_xp_tp_{yt} - pw_{2,tt} + pw_{6,t} + yu_xp_{tt}q_y
    - (u_xu_{xy} + u_{xt})p_tq \\
  &+ (u_x^2 + yu_{xt})p_tq_y - \frac{1}{2}yu_{xt}q_xw_5
    - \frac{1}{2}yu_{xt}q_xw_{1,t} + \frac{1}{2}yv_{xx}p_tw_{1,t}
    + \frac{1}{2}yv_{xx}p_tw_5\\
  &- \frac{1}{2}yv_{xt}p_xw_{1,t} - \frac{1}{2}yv_{xt}p_xw_5
    - y (u_xv_{xt} - u_{xt}v_x)p_tp_x + \frac{1}{2}yq_{yt}w_5
    + \frac{1}{2}yq_{yt}w_{1,t} \\
  &- \frac{1}{2}v_{xy}pw_5 - \frac{1}{2}v_{xy}pw_{1,t} - (u_xu_{xy} +
    u_{xt})pq_t
    + (yu_{xx} - 1)p_tw_{2,t} - (yu_{xx} - 1)p_tw_6\\
  &+ \frac{1}{2}u_xq_yw_5 - \frac{1}{2}qw_{1,tt} + \frac{1}{2}u_xq_yw_{1,t}
    + (u_xv_{xy} - u_{xy}v_x - v_{xt})pp_t + \frac{1}{2}(yu_{xx} - 1)q_tw_5\\
  &+ \frac{1}{2}v_xp_yw_{1,t} + \frac{1}{2}yq_yw_{1,tt}
    + (u_x^2 + yu_{xt})p_yq_t + \frac{1}{2}yq_yw_{5,t} - \frac{1}{2}u_{xy}qw_5
    - u_x p_{tt}q \\
  &- \frac{1}{2}u_{xy}qw_{1,t} + \frac{1}{2}v_xp_yw_5 - yv_{xt}p_tp_y
    + yu_xp_tq_{yt} - yu_{xt}p_xw_{2,t} + yu_{xt}p_xw_6\\
  &+ 2u_x(yu_{xx} - 1)p_tq_t - \frac{1}{2}qw_{5,t}
    + \frac{1}{2}(yu_{xx} - 1)q_tw_{1,t}),\\
  \phi_6^{w_5}
  &=\frac{1}{2}\dot{f}(2u_xp_tp_y + p_yw_5 + p_yw_{1,t})\\
  &- f(-2u_xu_{xt}p_tp_x + u_{xx}p_tw_5 + u_{xx}p_tw_{1,t} - u_{xt}p_xw_5
    - u_{xt}p_tp_y - u_xp_tp_{yt} - u_xp_{tt}p_y \\
  &- u_{xt}p_xw_{1,t} - \frac{1}{2}p_{yt}w_5 - \frac{1}{2}p_yw_{5,t}
    - \frac{1}{2}p_yw_{1,tt} - \frac{1}{2}p_{yt}w_{1,t}),\\
  \phi_6^{w_6}&=\frac{1}{2}\dot{f}(-2u_xp_tq_y - 2u_xp_y q_t + 2v_xp_tp_y
                + 2p_yw_6 - q_yw_5 - 2p_yw_{2,t} - q_yw_{1,t})\\
  &+ f(-u_xp_yq_{tt} - u_xp_{yt}q_t - p_{yt}w_{2,t} - p_yw_{2,tt} + p_yw_{6,t}
    + p_{yt}w_6 + \frac{1}{2}v_{xt}p_xw_{1,t}\\
  &+ u_xu_{xt}p_tq_x - u_{xt}p_xw_6 - u_{xt}p_yq_t + u_{xt}p_xw_{2,t}
    + u_{xx}p_tw_6 - u_{xx}p_tw_{2,t}\\
  &+ v_{xt}p_tp_y + v_xp_tp_{yt} + v_xp_{tt}p_y + u_xu_{xt} p_xq_t
    - 2u_{xx}u_xp_tq_t + \frac{1}{2}u_{xt}q_xw_{1,t} \\
  &- \frac{1}{2}u_{xx}q_tw_5 + (u_xv_{xt} - u_{xt}v_x)p_tp_x - u_xp_{tt}q_y
    - u_{xt}p_tq_y + \frac{1}{2}v_{xt}p_xw_5\\
  &- \frac{1}{2}v_{xx}p_tw_{1,t} + \frac{1}{2}u_{xt}q_xw_5- u_xp_tq_{yt}
    - \frac{1}{2}v_{xx}p_tw_5 - \frac{1}{2}u_{xx}q_tw_{1,t} -
    \frac{1}{2}q_{yt}w_5\\
  &- \frac{1}{2}q_{yt}w_{1,t} - \frac{1}{2}q_yw_{1,tt} -
    \frac{1}{2}q_yw_{5,t});\\
  \phi_7^{w_i}
  &=0,\quad i = 1,\dots,6;\\
  \phi_8^{w_i}
  &=0\quad i = 1,\dots,6;\\
  \phi_9^{w_i}
  &=0\quad i = 1,\dots,6;\\
  \phi_{10}^{w_i}
  &=0\quad i = 1,\dots,6;\\
  \phi_{11}^{w_i}
  &=0\quad i = 1,\dots,6;\\
  \phi_{12}^{w_1}
  &=0\quad i = 1,\dots,6,\\
  \phi_{12}^{w_2}
  &=\frac{1}{2}f(-2u_x^2p_tp_x + 2u_xp_tp_y - u_xp_xw_{1,t} - u_xp_xw_5
    + p_yw_5 + p_yw_{1,t}),\\
  \phi_{12}^{w_3}
  &=0,\\
  \phi_{12}^{w_4}
  &=\frac{1}{2}\dot{f}(2yu_xp_tp_y + 2u_xpp_t + yp_yw_5 + yp_yw_{1,t} - pw_5 -
    pw_{1,t})\\ 
  &+ f((u_x^2 + yu_{xt})p_tp_y + (u_xu_{xy} + u_{xt})pp_t
    + \frac{1}{2}(yu_{xx} - 1)p_tw_{1,t} + \frac{1}{2}u_xp_yw_{1,t}
    - \frac{1}{2}u_{xy}pw_{1,t}\\
  &+ \frac{1}{2}(yu_{xx} - 1)p_tw_5 - yu_{xt}u_xp_tp_x
    - \frac{1}{2}yu_{xt}p_xw_5 - \frac{1}{2}yu_{xt}p_xw_{1,t} + yu_xp_tp_{yt}
    + yu_xp_{tt}p_y\\
  &- \frac{1}{2}u_{xy}pw_5 + u_xpp_{tt} + \frac{1}{2}u_xp_yw_5
    + \frac{1}{2}yp_{yt}w_5 + \frac{1}{2}yp_yw_{5,t} + \frac{1}{2}yp_yw_{1,tt}\\
  &+ \frac{1}{2}yp_{yt}w_{1,t} - \frac{1}{2}pw_{5,t} - \frac{1}{2}pw_{1,tt}),\\
  \phi_{12}^{w_5}
  &=0,\\
  \phi_{12}^{w_6}
  &=\frac{1}{2}\dot{f}(2u_xp_tp_y + p_yw_5 + p_yw_{1,t})\\
  &+ \frac{1}{2}f(-2u_xu_{xt}p_tp_x + u_{xx}p_tw_5 + u_{xx}p_tw_{1,t}
    - u_{xt}p_xw_5 + 2u_{xt}p_tp_y + 2u_xp_tp_{yt} + 2u_xp_{tt}p_y\\
  &- u_{xt}p_xw_{1,t} + p_{yt}w_5 + p_yw_{5,t} + p_yw_{1,tt} + p_{yt}w_{1,t});\\
  \phi_{13}^{w_i}
  &=0\quad i = 1,\dots,6;\\
  \phi_{14}^{w_1}
  &=0\quad i = 1,\dots,6,\\
  \phi_{14}^{w_2}
  &=f(u_x^2p_xq_t - u_xv_xp_tp_x + u_xp_xw_{2,t} - u_xp_xw_6 + v_xp_tp_y
    - u_xp_yq_t - p_yw_{2,t} + p_yw_6),\\
  \phi_{14}^{w_3}
  &=0,\\
  \phi_{14}^{w_4}
  &=\dot{f}(yv_xp_tp_y - yu_xp_yq_t + v_xpp_t + u_xpq_t + yp_yw_6 -
    yp_yw_{2,t} - pw_6 + pw_{2,t})\\
  &- f(yu_{xt}v_xp_tp_x + yp_{yt}w_{2,t} + yp_yw_{2,tt} - yp_yw_{6,t}
    - yp_{yt}w_6 + yu_xp_yq_{tt} + yu_xp_{yt}q_t\\
  &- yv_xp_tp_{yt} - u_xp_yw_6 + u_xp_yw_{2,t} - u_{xy}pw_{2,t} - u_xpq_{tt}
    - v_xpp_{tt} + u_{xy}pw_6 - yv_xp_{tt}p_y\\
  &- yu_xu_{xt}p_xq_t - pw_{2,tt} + pw_{6,t} - (u_xv_x + yv_{xt})p_tp_y
    - (u_xu_{xy} + u_{xt})pq_t\\
  &+ (yu_{xx} - 1) p_tw_{2,t} - (yu_{xx} - 1)p_tw_6 - yu_{xt}p_xw_{2,t}
    - (u_{xy}v_x + v_{xt})p p_t + (u_x^2 + yu_{xt})p_yq_t\\
  &+ yu_{xt}p_xw_6 + u_x(yu_{xx} - 1)p_tq_t),\\
  \phi_{14}^{w_5}
  &=0,\\
  \phi_{14}^{w_6}
  &=\dot{f}(v_xp_tp_y - u_xp_yq_t + p_yw_6 - p_yw_{2,t})\\
  &- f(-v_xp_tp_{yt} - p_yw_{6,t} - v_{xt}p_tp_y + u_{xt}p_xw6 - u_{xx}p_tw_6
    - p_{yt}w_6 + u_xp_yq_{tt} - u_{xt}p_xw_{2,t}\\
  &+ p_yw_{2,tt} - u_xu_{xt}p_xq_t + u_{xt}v_xp_tp_x + u_{xx}u_xp_tq_t
    + p_{yt}w_{2,t}\\
  &- v_xp_{tt}p_y + u_xp_{yt}q_t + u_{xt}p_yq_t + u_{xx}p_tw_{2,t}).
\end{align*}

\end{document}